\documentclass[a4paper,UKenglish,cleveref,autoref,thm-restate]{lipics-v2021}

\newif\iflong
\newif\ifshort

\longtrue

\iflong
\else
\shorttrue
\fi

\newcommand\blfootnote[1]{%
  \begingroup
  \renewcommand\thefootnote{}\footnote{#1}%
  \addtocounter{footnote}{-1}%
  \endgroup
}

\usepackage{booktabs}
\usepackage{xspace}
\usepackage[obeyFinal,colorinlistoftodos,textsize=tiny]{todonotes}
\usepackage[normalem]{ulem}

\ifshort\title{An Algorithmic Framework for Locally Constrained Homomorphisms}\fi

\iflong\title{An Algorithmic Framework for Locally Constrained Homomorphisms}\fi 

\author{Laurent Bulteau}{LIGM, CNRS, Université Gustave Eiffel, France}{laurent.bulteau@univ-eiffel.fr}{https://orcid.org/0000-0003-1645-9345}{}
\author{Konrad K. Dabrowski}{Newcastle University, UK}{konrad.dabrowski@newcastle.ac.uk}{https://orcid.org/0000-0001-9515-6945}{Supported by EPSRC grant EP/V00252X/1.}
\author{Noleen K\"ohler}{Université Paris-Dauphine, PSL University, CNRS UMR7243, LAMSADE, Paris, France}{noleen.kohler@dauphine.psl.eu}{}{}
\author{Sebastian Ordyniak}{University of Leeds, UK}{sordyniak@gmail.com}{https://orcid.org/0000-0002-1825-0097}{Supported by EPSRC grant EP/V00252X/1.}
\author{Dani\"el Paulusma}{University of Durham, UK}{daniel.paulusma@durham.ac.uk}{https://orcid.org/0000-0001-5945-9287}{Supported by Leverhulme Trust grant RPG-2016-258.}

\authorrunning{L. Bulteau, K.\,K. Dabrowski, N. K\"ohler, S. Ordyniak, D. Paulusma} 

\Copyright{Laurent Bulteau, Konrad K. Dabrowski, Noleen K\"ohler, Sebastian Ordyniak, Dani\"el Paulusma} 



\relatedversion{} 
\nolinenumbers 

\hideLIPIcs

\newcommand{\N}{\mathbb{N}}
\newcommand{\Z}{\mathbb{Z}}
\newcommand{\SB}{\{\,}
\newcommand{\SM}{\;{:}\;}
\newcommand{\SE}{\,\}}

\newcommand{\bigoh}{\mathcal{O}}

\newcommand{\Nat}{\mathbb{N}}
\newcommand{\TTT}{\mathcal{T}}
\newcommand{\AAA}{\mathcal{A}}
\newcommand{\III}{\mathcal{I}}

\newcommand{\EEE}{\mathcal{E}}

\newcommand{\TD}{\mathcal{D}}

\newcommand{\tw}{{\textup{tw}}}
\newcommand{\td}{{\textup{td}}}
\newcommand{\vcn}{{\textup{vc}}}
\newcommand{\fr}{{\textup{fr}}}
\newcommand{\fvn}{{\textup{fv}}}
\newcommand{\dsnc}{{\textup{ds}_c}}
\newcommand{\pw}{{\textup{pw}}}
\newcommand{\drm}{\textup{drm}}

\newcommand{\undC}{C}
\newcommand{\fproj}[2]{#1|_{#2}}
\newcommand{\ccmap}{\textup{tc}}

\newcommand{\wSCMP}[4]{\textup{wSM}(#1,#2,#3,#4)}
\newcommand{\SCMP}[4]{\textup{SM}(#1,#2,#3,#4)}
\newcommand{\BCMP}[4]{\textup{BM}(#1,#2,#3,#4)}
\newcommand{\wSCMPn}{\textup{wSM}\xspace}
\newcommand{\SCMPn}{\textup{SM}\xspace}
\newcommand{\BCMPn}{\textup{BM}\xspace}

\newcommand{\ICMP}{\textup{ICM}\xspace}
\newcommand{\CH}{\textup{CH}\xspace}
\newcommand{\CON}{\textup{CON}\xspace}

\newcommand{\fullc}{\ensuremath\xrightarrow{_B}\xspace}
\newcommand{\surjc}{\ensuremath\xrightarrow{_S}\xspace}
\newcommand{\parc}{\ensuremath\xrightarrow{_I}\xspace}

\newcommand{\EXT}{\textup{Ext}}
\iffalse
\newcommand{\Hom}{\textsc{Homomorphism}}
\newcommand{\LBHOM}{\textsc{Locally Bijective Homomorphism}}
\newcommand{\LIHOM}{\textsc{Locally Injective Homomorphism}}
\newcommand{\LSHOM}{\textsc{Locally Surjective Homomorphism}}
\newcommand{\xLBHom}[1]{\textsc{Locally Bijective $#1$-Homomorphism}}
\newcommand{\xLIHom}[1]{\textsc{Locally Injective $#1$-Homomorphism}}
\newcommand{\xLSHom}[1]{\textsc{Locally Surjective $#1$-Homomorphism}}
\else
\newcommand{\Hom}{{\sc{Hom}}}
\newcommand{\LBHOM}{{\sc{LBHom}}}
\newcommand{\LIHOM}{{\sc{LIHom}}}
\newcommand{\LSHOM}{{\sc{LSHom}}}
\newcommand{\xLBHom}[1]{\textsc{$#1$-LBHom}}
\newcommand{\xLIHom}[1]{\textsc{$#1$-LIHom}}
\newcommand{\xLSHom}[1]{\textsc{$#1$-LSHom}}
\fi

\newcommand{\hy}{\hbox{-}\nobreak\hskip0pt}
\newcommand{\cc}[1]{{\mbox{\textnormal{\textsf{#1}}}}\xspace}  

\newcommand{\NP}{\cc{NP}}

\newcommand{\FPT}{\cc{FPT}}
\newcommand{\XP}{\cc{XP}}
\newcommand{\Weft}{{\cc{W}}}
\newcommand{\W}[1]{{\Weft}{{\normalfont{[#1]}}}}
\newcommand{\paraNP}{\cc{paraNP}}

\begin{document}

\maketitle

\begin{abstract}
A homomorphism $f$ from a guest graph $G$ to a host graph $H$ is locally bijective, injective or surjective if for every $u\in V(G)$, the restriction of $f$ to the neighbourhood of $u$ is bijective, injective or surjective, respectively. The corresponding decision problems, \LBHOM{}, \LIHOM{} and \LSHOM{}, are well studied both on general graphs and on special graph classes.  Apart from complexity results when the problems are parameterized by the treewidth and maximum degree of the guest graph, the three problems still lack a thorough study of their parameterized complexity. This paper fills this gap: we prove a number of new \FPT, \W{1}-hard and para-\NP-complete results by considering a hierarchy of parameters of the guest graph $G$. For our \FPT\ results, we do this through the development of a new algorithmic framework that involves a general ILP model. To illustrate the applicability of the new framework, we also use it to prove \FPT\ results for the {\sc Role Assignment} problem, which originates from social network theory and is closely related to locally surjective homomorphisms.
\ifshort
\blfootnote{\noindent \emph{Statements where proofs or details are provided in the appendix are marked with a $\star$.}}
\fi
\end{abstract}

\section{Introduction}\label{s-intro}

A {\it homomorphism} from a graph $G$ to a graph $H$ is a mapping $\phi: V(G) \to V(H)$  such that $\phi(u)\phi(v)\in E(H)$ for every $uv\in E(G)$. Graph homomorphisms generalise graph colourings (let $H$ be a complete graph) and have been intensively studied over a long period of time, both from a structural and an algorithmic perspective. We refer to the textbook of Hell and Ne\v{s}et\v{r}il~\cite{HN04} for a further introduction.

We write $G\to H$ if there exists a homomorphism from $G$ to $H$; here, $G$ is called the {\it guest graph} and $H$ is the {\it host graph}.
We denote the corresponding decision problem by  \Hom{}, and if $H$ is fixed, that is, not part of the input, we write $H$-\Hom{}. The renowned Hell-Ne\v{s}et\v{r}il dichotomy~\cite{HN90}  states that  $H$-\Hom{} is polynomial-time solvable if $H$ is bipartite, and \NP-complete otherwise.
We denote the vertices of $H$ by $1,\ldots,|V(H)|$ and call them {\it colours}.
\iflong The reason for doing this is that graph homomorphisms generalise graph colourings: there exists a homomorphism from a graph $G$ to the complete graph on $k$ vertices if and only if $G$ is $k$-colourable.\fi

Instead of fixing the host graph~$H$, one can also restrict the
structure of the guest graph~$G$ by bounding some graph
parameter. \ifshort Here it is known that, if $\FPT\neq\W{1}$, then
\Hom{}  can be solved in polynomial time if and only if the so-called
core of the guest graph has bounded treewidth~\cite{Gr07}.\fi
\iflong A classical result states that \Hom{} is polynomial-time solvable when the guest graph~$G$ has bounded treewidth~\cite{CR00,Fr90}. The {\it core} of a graph $G$ is the subgraph $F$ of $G$ such that $G\to F$ and there is no proper subgraph $F'$ of $F$ with $G\to F'$
(the core is unique up to isomorphism~\cite{HN92}).
Dalmau, Kolaitis and Vardi~\cite{DKV02} proved that the \Hom{} problem is polynomial-time solvable even if the core of the guest graph~$G$ has bounded treewidth. This result was strengthened by Grohe~\cite{Gr07}, who proved that if $\FPT\neq\W{1}$, then \Hom{}  can be solved in polynomial time if and only if this condition holds.
\fi

\medskip
\noindent
{\bf Locally constrained homomorphisms.} We are interested in three well-studied variants of graph homomorphisms that occur after placing constraints on the neighbourhoods of the vertices of the guest graph $G$.
Consider a homomorphism $\phi$ from a graph $G$ to a graph~$H$.
We say that $\phi$ is locally injective, locally bijective or locally surjective for $u \in V(G)$ if
the restriction $\phi_u$ to the neighbourhood 
$N_G(u)=\{v\; |\; uv\in E(G)\}$ 
of $u$ is injective, bijective or surjective. We say that
$\phi$ is {\it locally injective}, {\it locally bijective}  or {\it locally surjective} if it is locally injective, locally bijective, or locally surjective for every $u \in V(G)$.
We denote these {\it locally constrained} homomorphisms by
$G\fullc H$, $G\parc H$ and $G\surjc H$, respectively.

The three variants have been well studied in several settings over a long period of time. For example, locally injective homomorphisms are also known as {\it partial graph coverings} and are used in telecommunications~\cite{FK02}, in distance constrained labelling~\cite{FKK01} and as indicators of the existence of homomorphisms  of  derivative graphs~\cite{Ne71}. Locally bijective homomorphisms originate from topological graph theory~\cite{Bi74,Ma67} and are more commonly known as {\it graph coverings}. They are used in distributed
computing~\cite{An80,AG81,Bo89} and in constructing highly transitive regular graphs~\cite{Bi82}. Locally surjective homomorphisms are sometimes called {\it colour dominations}~\cite{KT00}. They have applications in distributed computing~\cite{CMZ06,CP11} and in social science~\cite{EB91,PR01,RS01,WR83}. In the latter context they are known as {\it role assignments}, as we will explain in more detail below.

Let  \LBHOM{}, \LIHOM{} and \LSHOM{} be the three problems of deciding, for two graphs $G$ and $H$, whether $G\fullc H$, $G\parc H$ or $G\surjc H$ holds, respectively. As before, we write
 \xLBHom{H}, \xLIHom{H} and \xLSHom{H} in the case where the host graph~$H$ is fixed.
Out of the three problems, only the complexity of \xLSHom{H} has been completely classified, both for general graphs and bipartite graphs~\cite{FP05}. We refer to a series of papers~\cite{AFS91,BLT11,FK02,FKP08,KPT97,KPT98,LT10} for polynomial-time solvable and \NP-complete cases of \xLBHom{H} and \xLIHom{H}; see also the survey by Fiala and Kratochv\'{i}l~\cite{FK08}. Some more recent results include sub-exponential algorithms for \xLBHom{H}, \xLIHom{H} and \xLSHom{H} on string graphs~\cite{OR20} and complexity results for \xLBHom{H} for host graphs $H$ that are multigraphs~\cite{KTT16} or that have semi-edges~\cite{BFHJK}.

In our paper we assume that both $G$ and $H$ are part of the input.
We note a fundamental difference between locally injective homomorphisms on one hand and locally bijective and surjective homomorphisms on the other hand. Namely, for connected graphs $G$ and $H$, we must have
 $|V(G)|\geq |V(H)|$ if $G\fullc H$ or $G\surjc H$, whereas $H$ might be arbitrarily larger than $G$ if $G\parc H$ holds.  For example, if we let $G$ be a complete graph, then $G\parc H$ holds if and only if $H$ contains a clique on at least $|V(G)|$ vertices.

The above difference is also reflected in the complexity results for the three problems under input restrictions.
In fact, \LIHOM{} is closely related to the {\sc Subgraph Isomorphism} problem and is usually the
hardest problem. For example, \LBHOM{} is {\sc Graph Isomorphism}-complete on chordal guest graphs, but polynomial-time solvable on interval
guest
graphs and \LSHOM{} is \NP-complete on chordal guest graphs, but polynomial-time solvable on proper interval
guest
graphs~\cite{HHP12}.
In contrast, \LIHOM{} is \NP-complete even on complete guest graphs~$G$, which follows from a reduction from the {\sc Clique} problem via the aforementioned equivalence: $G\parc H$ holds if and only if $H$ contains a clique on at least $|V(G)|$ vertices.

To give another example,
\LBHOM{}, \LSHOM{} and \LIHOM{} are \NP-complete for guest graphs~$G$ of path-width at most $5$, $4$ and $2$, respectively~\cite{CFHPT15} (all three problems are polynomial-time solvable if $G$ is a tree~\cite{CFHPT15,FP10}).
Note that these hardness results imply that the aforementioned polynomial-time result on \Hom{} for guest graphs~$G$ of bounded treewidth~\cite{CR00,Fr90} does not carry over to any of the three locally constrained homomorphism problems. It is also known that \LBHOM{}~\cite{Kr94}, \LSHOM{}~\cite{KT00} and
\LIHOM{}~\cite{FK02} are \NP-complete even if $G$ is cubic and $H$ is the complete graph $K_4$ on four vertices, but
polynomial-time solvable if~$G$ has bounded treewidth and one of the two graphs~$G$ or~$H$ has bounded maximum degree~\cite{CFHPT15}.

\medskip
\noindent
{\bf An Application.} Locally surjective homomorphisms from a graph $G$ to a graph $H$ are known as $H$-role assignments in social network theory. We will include this topic in our investigation and provide some brief context.
Suppose we are given a social network of individuals whose properties we aim to characterise. Can we assign each individual a role such that individuals with the same role relate in the same way to other individuals with some role, using exactly $h$ different roles in total? 
To formalise this question, we model the network as a graph~$G$, where vertices represent individuals and edges represent the existence of a relationship between two individuals.
We now ask whether $G$ has
an {\it $h$-role assignment}, that is, a function $f$ that assigns each vertex $u\in V(G)$ a {\it role} $f(u)\in \{1,\ldots,h\}$, such that $f(V(G))=\{1,\ldots,h\}$ and for every two vertices $u$ and~$v$, if $f(u)=f(v)$ then $f(N_G(u))=f(N_G(v))$. 

Role assignments were introduced
by White and Reitz~\cite{WR83} as {\it regular equivalences} and were called {\it role colourings} by Everett and Borgatti~\cite{EB91}.
We observe that two adjacent vertices $u$ and $v$ may have the same role, that is, $f(u)=f(v)$ is allowed (so role assignments are not proper colourings).
Hence, a connected graph $G$ has an $h$-role assignment if and only if $G\surjc H$ for some connected graph $H$ with $|V(H)|=h$, as long as we allow $H$ to have self-loops (while we assume that $G$ is a graph with no self-loops).

The {\sc Role Assignment} problem is to decide, for a graph $G$ and an integer $h$, whether $G$ has an $h$-role assignment. If $h$ is fixed, we denote the problem $h$-{\sc Role Assignment}. Whereas {\sc $1$-Role Assignment} is trivial, {\sc $2$-Role Assignment} is \NP-complete~\cite{RS01}. In fact, {\sc $h$-Role Assignment} is \NP-complete for planar graphs $(h\geq 2)$~\cite{PR15}, cubic graphs $(h\geq 2)$~\cite{PR}, bipartite graphs $(h\geq 3)$~\cite{PS}, chordal graphs $(h\geq 3)$~\cite{HPR10} and split graphs $(h\geq 4)$~\cite{Do16}.
 Very recently, Pandey, Raman and Sahlo~\cite{PRS21} gave an $n^{\bigoh(h)}$-time algorithm for {\sc Role Assignment} on general graphs and an $f(h)n^{\bigoh(1)}$-time algorithm on forests.

\medskip
\noindent
{\bf Our Focus.} We continue the line of study in~\cite{CFHPT15} and focus on the following research question:

\medskip
\noindent
{\it For which parameters of the guest graph do \LBHOM{}, \LSHOM{} and \LIHOM{} become fixed-parameter tractable?}

\medskip
\noindent
We will also apply our new techniques towards answering this question for the {\sc Role Assignment} problem.
In order to address our research question, we need some additional terminology.
A graph parameter $p$ {\it dominates} a parameter~$q$ if there is a function~$f$ such that $p(G)\leq f(q(G))$ for every graph~$G$.
If $p$ dominates $q$ but $q$ does not dominate $p$, then $p$ is {\it more powerful (less restrictive)} than $q$. We denote this by $p \rhd q$.
If neither $p$ dominates $q$ nor $q$ dominates $p$, then $p$ and $q$ are {\it incomparable (orthogonal)}.
Given the para-\NP-hardness results on \LBHOM{}, \LSHOM{} and \LIHOM{} for graph classes of bounded path-width~\cite{CFHPT15}, we will consider a range of graph parameters that are less powerful than path-width. In this way we aim to increase our understanding of the (parameterized) complexity of  \LBHOM{}, \LSHOM{} and \LIHOM{}.

For an integer $c\geq 1$, a {\it $c$-deletion set} of a graph $G$ is a subset $S\subseteq V(G)$ such that every connected component of $G\setminus S$ has at most $c$ vertices. The {\it  $c$-deletion set number} $\dsnc(G)$ of a graph $G$ is the minimum size of a $c$-deletion set in $G$. If $c=1$, then we obtain the {\it vertex cover number} $\vcn(G)$ of $G$.
The $c$-deletion set number is also known as {\it vertex integrity}~\cite{DDH16}.
It is closely related to the {\it fracture number} $\fr(G)$, 
introduced in~\cite{DEGKO17}, which is the minimum $k$ such that $G$ has a $k$-deletion set on at most~$k$ vertices.
Note that $\fr(G)\leq \max\{c,\dsnc(G)\}$ holds for every integer~$c$.
The {\it feedback vertex set number} $\fvn(G)$ of a graph $G$  is the size of a smallest set $S$ such that $G\setminus S$ is a forest.
We write $\tw(G)$, $\pw(G)$ and $\td(G)$ for the treewidth, path-width and tree-depth of a graph $G$, respectively; see~\cite{NOdM12} for more information,
in particular on tree-depth. It is known that
\ifshort$\tw(G) \rhd \pw(G) \rhd \td(G) \rhd \fr(G) \rhd \dsnc(G)
(\mbox{fixed}\; c) \rhd \vcn(G) \rhd |V(G)|,$ \fi
\iflong $$\tw(G) \rhd \pw(G) \rhd \td(G) \rhd \fr(G) \rhd \dsnc(G)
(\mbox{fixed}\; c) \rhd \vcn(G) \rhd |V(G)|,$$ \fi where the second relationship is proven
in~\cite{BGHK95} and the others follow immediately from their definitions (see also Section~\iflong\ref{ssec:pre-graph}\fi\ifshort\ref{sec:pre}\fi).
It is readily seen that \ifshort $\tw(G)\rhd \fvn(G) \rhd
\mbox{ds}_2(G)$ \fi \iflong $$\tw(G)\rhd \fvn(G) \rhd
\mbox{ds}_2(G)$$ \fi and that $\fvn(G)$ is incomparable with the parameters $\pw(G)$, $\td(G)$, $\fr(G)$ and $\dsnc(G)$ for every fixed
$c\geq 3$
(consider e.g. a tree of \iflong arbitrarily \fi large path-width and
the disjoint union of \iflong arbitrarily \fi many triangles).

\begin{table}
\begin{center}\resizebox{\textwidth}{!}{
\begin{tabular}{llll}
\toprule
guest graph parameter &\LIHOM{} & \LBHOM{} & \LSHOM{} \\
\midrule
$|V(G)|$ &\XP, \W{1}-hard~\cite{DF95} &\FPT &\FPT\\
vertex cover number &\textcolor{blue}{\XP} (Theorem~\ref{the:IHvc}), \W{1}-hard &\FPT &\FPT\\
$c$-deletion set number 
(fixed $c$)
 &\textcolor{blue}{para-\NP-c $(c\geq 2)$} (Theorem~\ref{the:IHdic}) & \FPT &\FPT\\
fracture number & para-\NP-c &\textcolor{blue}{\FPT} (Theorem~\ref{thm:FPTalgoLSHOMandLBHOM}) &\textcolor{blue}{\FPT} (Theorem~\ref{thm:FPTalgoLSHOMandLBHOM})\\
tree-depth &para-\NP-c &\textcolor{blue}{para-\NP-c} (Theorem~\ref{t-np}) &\textcolor{blue}{para-\NP-c} (Theorem~\ref{t-np})\\
path-width &para-\NP-c~\cite{CFHPT15} &para-\NP-c~\cite{CFHPT15} &para-\NP-c~\cite{CFHPT15}\\
treewidth&para-\NP-c &para-\NP-c &para-\NP-c\\
maximum degree &para-\NP-c~\cite{FK02} &para-\NP-c~\cite{Kr94} &para-\NP-c~\cite{KT00}\\
treewidth plus maximum degree &\XP, \W{1}-hard &\XP~\cite{CFHPT15}  &\XP~\cite{CFHPT15}\\
feedback vertex set number &para-\NP-c &\textcolor{blue}{para-\NP-c} (Theorem~\ref{t-np2}) &\textcolor{blue}{para-\NP-c} (Theorem~\ref{t-np2})
 \end{tabular}}

\vspace*{2.5mm}
\caption{Table of results. The results in blue are the new results proven in this paper. The results in black are either known results, some of which are now also implied by our new results, or follow immediately from other results in the table;
in particular, for a graph $G$, $\dsnc(G)\geq \fr(G)$ if $c\leq \fr(G)-1$, and
$\dsnc(G)\leq \fr(G)$ if $c\geq \fr(G)$.
Also note that \LIHOM{} is \W{1}-hard when parameterized by $|V(G)|$, as {\sc Clique} is \W{1}-hard when parameterized by the clique number~\cite{DF95}, so as before, we can let $G$ be the complete graph in this case.}\label{t-thetable}
\vspace*{-1cm}
\end{center}
\end{table}

\medskip
\noindent
{\bf Our Results.}
We prove a number of new parameterized complexity results for \LBHOM{}, \LSHOM{} and \LIHOM{} by taking some property of the guest graph~$G$ as the parameter. In particular, we consider the graph parameters above.
Our two main results, which are proven in Section~\ref{sec:appl}, show that \LBHOM{} and \LSHOM{} are fixed-parameter tractable
parameterized by the fracture number of $G$.
These two results cannot be strengthened to the tree-depth of the guest graph, for which we prove para-\NP-completeness in Section~\ref{s-npcom}. Note that the latter results imply the known para-\NP-completeness results for path-width of the guest graph~\cite{CFHPT15}. 
In Section~\ref{s-npcom} we also prove that  \LBHOM{} and \LSHOM{} are para-\NP-complete when parameterized by the feedback vertex set number of the guest graph. This result and the para-NP-hardness for tree-depth motivated us to consider the fracture number as a natural remaining graph parameter for obtaining an fpt algorithm. 

The above para-\NP-completeness results for \LBHOM{} in fact even hold for {\sc $3$-FoldCover},
the restriction of \LBHOM{} to input pairs $(G,H)$ where $|V(G)|=3|V(H)|$. The {\sc $k$-FoldCover} problem was introduced in~\cite{Bo89} (where it was called the $k$-{\sc Graph Covering} problem). 
In fact, the aforementioned result of~\cite{CFHPT15} on \LBHOM{} for path-width is the first proof that $3$-{\sc FoldCover} is \NP-complete, and the proof can easily be adapted for $k\geq 4$, as observed by Klav\'ik~\cite{Kl17}.

In Section~\ref{s-injective} we prove that \LIHOM{} is in \XP\ and \W{1}-hard when parameterized by the vertex cover number, or equivalently, the $c$-deletion set number for $c=1$. We then show that the \XP-result for \LIHOM{} cannot be generalised to hold for $c\geq 2$.
In fact, in the same section, we will determine the complexity of \LIHOM{} on graphs with $c$-deletion set number at most $k$ for every fixed pair of integers $c$ and $k$. 
Our results for \LBHOM{}, \LSHOM{} and \LIHOM{}  are summarised, together with the known results, in Table~\ref{t-thetable}.

\medskip
\noindent
{\bf Algorithmic Framework.}
The \FPT\ algorithms for \LBHOM{} and \LSHOM{} are proven via a new
algorithmic framework (described in detail in Section~\ref{s-algo})
that involves a reduction to an integer linear program (ILP) that has
a wider applicability. To illustrate this, in Section~\ref{sec:appl}
we also use our general framework to prove that {\sc Role Assignment}
is \FPT\ when parameterized by $c+\dsnc$.
\iflong We emphasize that in our framework the host graph~$H$ is not fixed,
but part of the input, in contrast to other frameworks that include
the locally constrained homomorphism problems (and that consequently
work for more powerful graph parameters), such as the framework of
locally checkable vertex partitioning problems~\cite{BTV13, TP97} or
the framework of Gerber and Kobler~\cite{GK03} based on (feasible) interval degree constraint matrices.
\fi

\medskip
\noindent
{\bf Techniques.}
The main ideas behind our algorithmic ILP framework are as
follows. Let $G$ and $H$ be the guest and host graphs, respectively.
First, we observe that if $G$ has a
$c$-deletion of size at most $k$ and there is a locally surjective
homomorphism from $G$ to $H$, then $H$ must also have a
$c$-deletion set of size at most $k$.
However it does not suffice to compute $c$-deletion sets
$D_G$ and $D_H$ for $G$ and $H$, guess a partial homomorphism $h$
from $D_G$ to $D_H$, and use the structural properties of $c$-deletion
sets to decide whether $h$ can be extended to
a desired homomorphism from $G$ to $H$. This is because
a homomorphism from $G$ to $H$ does not necessarily map $D_G$ to $D_H$.
Moreover, even if it did, vertices in $G\setminus D_G$ can still be mapped to vertices in $D_H$.
Consequently, components of $G\setminus D_G$ can still be mapped to more than one
component of $H\setminus D_H$. This makes it difficult to decompose the
homomorphism from $G$ to $H$ into small independent parts.
To overcome this challenge, we
prove that there are small sets $D_G$ and $D_H$ of vertices in $G$ and
$H$, respectively, such that every locally surjective homomorphism
from $G$ to $H$ satisfies: \ifshort(1) the pre-image of $D_H$ is a
subset of $D_G$, (2) $D_H$ is a $c'$-deletion set for $H$ for some $c'$ bounded in
terms of only $c+k$, and (3) all but at most $k$ components of $G\setminus
D_G$ have at most $c$ vertices and, while the remaining
components can be arbitrary large, their treewidth is bounded in terms
of $c+k$.
\fi\iflong\\[-11pt]
\begin{enumerate}
\item the pre-image of $D_H$ is a subset of $D_G$,
\item $D_H$ is a $c'$-deletion set for $H$ for some $c'$ bounded in
terms of only $c+k$, and 
\item all but at most $k$ components of $G\setminus
D_G$ have at most $c$ vertices and, while the remaining
components can be arbitrary large, their treewidth is bounded in terms
of $c+k$.\\[-11pt]
\end{enumerate}\fi
As $D_G$ and $D_H$ are small, we can enumerate all possible
homomorphisms
 from some subset of $D_G$ to~$D_H$. Condition~2 allows us to show that any locally
surjective homomorphism from $G$ to $H$ can be decomposed into locally
surjective homomorphisms from a small set of components of $G\setminus
D_G$ (plus $D_G$) to one component of $H\setminus D_H$ (plus $D_H$). This
enables us to formulate the question of whether 
a homomorphism from a subset of $D_G$ to a subset of $D_H$
can be extended to a desired homomorphism from $G$ to $H$ 
in terms of an ILP. Finally, Condition~3 allows us to efficiently
compute the possible parts of the decomposition, that is, which (small) sets of
components of $G\setminus D_G$ can be mapped to which components of $D_H$.

\section{Preliminaries}\label{sec:pre}

\iflong We use standard notation from graph theory, as can be found in e.g.~\cite{diestel00}. \fi
Let $G$ be a graph. We denote the vertex set and edge set of $G$ by $V(G)$ and $E(G)$, respectively. 
Let $X \subseteq V(G)$ be a set of vertices of $G$.
The \emph{subgraph of $G$ induced by $X$}, denoted $G[X]$, is the graph with vertex set $X$ and edge set $E(G) \cap [X]^2$.
Whenever the underlying graph is clear from the context, we will
sometimes refer to an induced subgraph simply by its set of vertices.
We use $G \setminus X$ to denote the subgraph of $G$ induced by $V(G) \setminus X$.
Similarly, for $Y \subseteq E(G)$ we let $G \setminus Y$ be the subgraph of $G$ obtained by deleting all edges in $Y$ from $G$.

For a graph $G$ and a vertex $u \in V(G)$, we let $N_G(u)=\{v\; |\; uv\in E(G)\}$ and
$N_G[v]=N_G(v)\cup \{v\}$ denote the open and closed neighbourhood of $v$ in $G$,
respectively.  We let $\Delta(G)$ be the maximum degree of $G$.
Recall that we assume that the guest graph~$G$ does not contain self-loops, while the host graph~$H$ is permitted to have self-loops.
In this case, by definition, $u\in N_H(u)$ if $uu\in E(H)$.

\iflong
\subsection{Parameterized Complexity}

In parameterized complexity~\cite{CyganFKLMPPS15,DowneyF13,FlumGrohe06},
the complexity of a problem is studied not only with respect to the
input size, but also with respect to some problem parameter(s). The
core idea behind parameterized complexity is that the combinatorial
explosion resulting from the \NP-hardness of a problem can sometimes
be confined to certain structural parameters that are small in
practical settings. We now proceed to the formal definitions.

A {\it parameterized problem} $Q$ is a subset of $\Omega^* \times
\mathbb{N}$, where $\Omega$ is a fixed alphabet. Each instance of $Q$ is a pair $(I, \kappa)$, where $\kappa \in \Nat$ is called the {\it
parameter}. A parameterized problem $Q$ is
{\it fixed-parameter tractable} (\FPT)~\cite{CyganFKLMPPS15,DowneyF13,FlumGrohe06}, if there is an
algorithm, called an {\em \FPT-algorithm},  that decides whether an input $(I, \kappa)$
is a member of $Q$ in time $f(\kappa) \cdot |I|^{\bigoh(1)}$, where $f$ is a computable function and $|I|$ is the size of the input instance. The class \FPT{} denotes the class of all fixed-parameter
tractable parameterized problems. 

A parameterized problem $Q$
is {\it \FPT-reducible} to a parameterized problem $Q'$ if there is
an algorithm, called an \emph{\FPT-reduction}, that transforms each instance $(I, \kappa)$ of $Q$
into an instance $(I', \kappa')$ of
$Q'$ in time $f(\kappa)\cdot |I|^{\bigoh(1)}$, such that $\kappa' \leq g(\kappa)$ and $(I, \kappa) \in Q$ if and
only if $(I', \kappa') \in Q'$, where $f$ and $g$ are computable
functions. By \emph{fpt-time}, we denote time of the form $f(\kappa)\cdot |I|^{\bigoh(1)}$, where $f$ is a computable function.
Based on the notion of \FPT-reducibility, a hierarchy of
parameterized complexity, {\it the \cc{W}-hierarchy} $=\bigcup_{t
\geq 0} \W{t}$, where $\W{t} \subseteq \W{t+1}$ for all $t \geq 0$, has
been introduced, in which the $0$-th level \W{0} is the class {\it
\FPT}. The notions of hardness and completeness have been defined for each level
\W{$i$} of the \cc{W}-hierarchy for $i \geq 1$ \cite{CyganFKLMPPS15,DowneyF13}.
It is commonly believed that $\W{1} \neq \FPT$ (see \cite{CyganFKLMPPS15,DowneyF13}).
This assumption has served as the main working hypothesis of fixed-parameter intractability.
The class \XP{} contains parameterized problems that can be solved in $\bigoh(|I|^{f(\kappa)})$ time, where $f$ is a computable function.
It contains the class \W{t}, for all $t \geq 0$, and every problem in \XP{} is polynomial-time solvable when the parameter is bounded by a constant.
The class \paraNP{} is the class of parameterized problems that can be solved by non-deterministic algorithms in $f(\kappa)\cdot |I|^{\bigoh(1)}$ time, where $f$ is a computable function.
A problem is \emph{\paraNP{}-hard} if it is \NP-hard for a constant value of the parameter~\cite{FlumGrohe06}.
\fi

\iflong\subsection{Graph Parameters}\label{ssec:pre-graph}\fi

A \emph{$(k,c)$-extended deletion set} for $G$ is a set $D \subseteq
V(G)$ such that: \ifshort (1) every component of $G\setminus D$ either has at most $c$ vertices or has a
  $c$-deletion set of size at most $k$ and (2)
at most $k$ components of $G\setminus D$ have more than $c$
vertices. \fi\iflong
\begin{itemize}
\item every component of $G\setminus D$ either has at most $c$ vertices or has a
  $c$-deletion set of size at most $k$ and
\item at most $k$ components of $G\setminus D$ have more than $c$ vertices.
\end{itemize}
\fi
 We need the following well-known fact\ifshort s\fi:
\begin{proposition}[\cite{DBLP:journals/ai/KroneggerOP19}]\label{pro:comp-dels}
  Let $G$ be a graph and let $k$ and $c$ be natural numbers. Then,
  deciding whether $G$ has a $c$-deletion set of size at most $k$ is
  fixed-parameter tractable parameterized by $k+c$.
\end{proposition}

\iflong
\subparagraph*{Tree-depth.} Tree-depth is closely related to treewidth, and the structure of graphs
of bounded tree-depth is well understood~\cite{NOdM12}.
A useful way of thinking about graphs of bounded tree-depth is that they are
(sparse) graphs with no long paths.

The \emph{tree-depth} of an undirected graph $G$, denoted by $\td(G)$,
is the smallest natural number $k$ such
that there is an undirected rooted forest $F$ with vertex set $V(G)$ of height
at most $k$ for which $G$ is a subgraph of $\undC(F)$, where
$\undC(F)$ is called the \emph{closure} of $F$ and is the undirected graph with vertex set $V(F)$ having an
edge between $u$ and $v$ if and only if $u$ is an ancestor of $v$ in
$F$. A forest $F$ for which $G$ is a subgraph of $\undC(F)$ is also
called a \emph{tree-depth decomposition}, whose \emph{depth} is equal
to the height of the forest plus one.
Informally a graph has tree-depth at most $k$ if it can be
embedded in the closure of a forest of height $k$. Note that if $G$ is
connected, then it can be embedded in the closure of a tree instead of
a forest. 

\subparagraph*{Treewidth.}
A \emph{tree-decomposition}~$\mathcal{T}$ of a graph $G$ is a pair 
$(T,\chi)$, where $T$ is a tree and $\chi$ is a function that assigns each 
tree node $t$ a set $\chi(t) \subseteq V(G)$ of vertices such that the following 
conditions hold:
\begin{enumerate}[(P1)]
\item For every
edge $uv \in E(G)$, there is a tree node
  $t$ such that $u,v\in \chi(t)$.
\item For every vertex $v \in V(G)$,
  the set of tree nodes $t$ with $v\in \chi(t)$ induces a non-empty subtree of~$T$.
\end{enumerate}

The sets $\chi(t)$ are called \emph{bags} of the decomposition~$\mathcal{T}$ and $\chi(t)$ 
is the bag associated with the tree node~$t$. 
The \emph{width} of a tree-decomposition 
$(T,\chi)$ is the size of a largest bag minus~$1$. The \emph{treewidth} of a graph $G$,
denoted by $\tw(G)$, is the minimum
width over all tree-decompositions of~$G$. 
\fi

\iflong\subparagraph*{Relationships and Properties.}
The following proposition summaries the known relationships between the parameters we consider.\fi

\iflong
\begin{proposition}[\cite{NOdM12}]
  \label{pro:parrel}
  Let $G$ be a graph and let $k$ and $c$ be natural numbers. Then:
  \begin{itemize}
  \item if $G$ has a $c$-deletion set of size at most $k$, then $\td(G)\leq
    k+c$.
  \item if $G$ has a $(k',c)$-extended deletion set of size at most $k$, then
    $\td(G)\leq k'+k+c$.    
  \item $\tw(G)\leq \td(G)$.
  \end{itemize}
\end{proposition} 
\fi
\ifshort
\begin{proposition}[\cite{NOdM12}]
  \label{pro:parrel}
  Let $G$ be a graph and $k$ and $c$ be natural numbers. Then:
  \begin{itemize}
  \item if $G$ has a $c$-deletion set of size at most $k$, then $\td(G)\leq
    k+c$.
  \item if $G$ has a $(k',c)$-extended deletion set of size at most $k$, then
    $\td(G)\leq k'+k+c$.    
  \item $\tw(G)\leq \td(G)$.
  \end{itemize}
\end{proposition} 
\fi

\iflong\subsection{Locally Constrained Homomorphisms}\fi\ifshort\subparagraph*{Locally Constrained Homomorphisms.}\fi

We always allow self-loops for the host graph, but not for the guest graph (see also Section~\ref{s-intro}).
Here we show some basic properties of locally constrained homomorphisms.

\iflong \begin{observation} \fi \ifshort \begin{observation}[$\star$] \fi\label{obs:surjective}
  Let $G$ and $H$ be non-empty connected graphs and let $\phi$ be a locally surjective homomorphism from $G$ to $H$.
  Then $\phi$ is surjective.
\end{observation}
\iflong
\begin{proof}
  Suppose not, and let $C$ be the set of vertices in $V(H)\setminus \phi(V(G))$.
  Note that $C\neq \emptyset$ (because otherwise $\phi$ is surjective) and $\phi(V(G))\neq \emptyset$ (because $G$ is non-empty).
  Because $H$ is connected, there is an edge $uv \in E(H)$ such that $u\in V(C)$ and $v \in \phi(V(G))$.
  But then, the mapping $\phi_x : N_G(x) \rightarrow N_H(v)$ is not surjective for any vertex $x \in \phi^{-1}(v)$.
\end{proof}
\fi

\iflong \begin{observation} 
  \label{obs:preimage}
  Let $G$ and $H$ be non-empty connected graphs with a homomorphism $\phi$ from $G$ to $H$ and let $I \subseteq \phi(V(G))$.
  Let $P=\phi^{-1}(I)$ and $\phi_R=\fproj{\phi}{P}$.
  If $\phi$ is a locally injective, surjective or bijective homomorphism, then $\phi_R$ is a locally injective, surjective or bijective homomorphism, respectively, from $G[P]$ to $H[I]$.
\end{observation}
\fi
\iflong
\begin{proof}
  Clearly, $\phi_R$ is a homomorphism from $G[P]$ to $H[I]$ and since $\phi_R$ is a restriction of $\phi$, it follows that if $\phi$ is locally injective, then so is $\phi_R$.
  It remains to show that if $\phi$ is locally surjective, then so is $\phi_R$.
  Suppose, for contradiction, that $\phi$ is locally surjective, but $\phi_R$ is not.
  Then there is a vertex $v \in P$ such that $\phi_R(N_G(v)\cap P) \subsetneq N_H(\phi_R(v)) \cap I$.
  However, since $\phi$ does not map any vertex in $V(G)\setminus P$ to a vertex of~$I$, it follows that $\phi(N_G(v)) \cap I \subsetneq	N_H(\phi(v)) \cap I$, so $\phi(N_G(v)) \neq N_H(\phi(v))$. Thus $\phi$ is not surjective, a contradiction.
\end{proof}
\fi

\iflong \begin{observation} \fi \ifshort \begin{observation}[$\star$] \fi\label{obs:separators}
  Let $G$ and $H$ be graphs, let $D \subseteq V(G)$, and let $\phi$ be a homomorphism from $G$ to $H$. 
  Then, for every component $C_G$ of $G\setminus D$ such that $\phi(C_G)\cap\phi(D)=\emptyset$, there is a component $C_H$ of $H\setminus \phi(D)$ such that
$\phi(C_G) \subseteq C_H$.
  Moreover, if $\phi$ is locally injective/surjective/bijective, then $\phi_R=\fproj{\phi}{D\cup C_G}$ is a homomorphism from $G'=G[D\cup C_G]$ to $H'=H[\phi(D)\cup C_H]$ that is locally injective/surjective/bijective for every  $v \in V(C_G)$.
\end{observation}
\iflong
\begin{proof}
  Suppose for a contradiction that this is not the case.
  Then, there is a component $C_G$ of $G\setminus D$ and an edge $uv \in E(C_G)$ such that $\phi(u)$ and $\phi(v)$ are in different components of $H\setminus \phi(D)$.
  Therefore, $\phi(u)\phi(v) \notin E(H)$, contradicting our assumption that $\phi$ is a homomorphism.

  Towards showing the second statement, first note that $\phi_R$ is a homomorphism from $G'$ to $H'$.
  Moreover, $N_G[v]=N_{G'}[v]$ for every vertex $v \in V(C_G)$, so if $\phi$ is locally injective/surjective/bijective for a vertex $v \in V(C_G)$, then so is $\phi_R$.
\end{proof}
\fi

\iflong The following lemma is a basic but crucial observation showing that if $G\surjc H$ and
$G$ has a small $c$-deletion set, then so does $H$.\fi
\iflong \begin{lemma} \fi \ifshort \begin{lemma}[$\star$] \fi\label{lem:delSet}
  Let $G$ and $H$ be non-empty connected graphs, let $D \subseteq
  V(G)$ be a $c$-deletion set for $G$, and let
  $\phi$ be a locally surjective homomorphism
  from $G$ to $H$.
  Then $\phi(D)$ is a $c$-deletion set for $H$.
\end{lemma}

\iflong
\begin{proof}
  Suppose not, then there is a component $C_H$ of $H\setminus \phi(D)$ such that $|C_H|>c$.
  By Observation~\ref{obs:surjective}, it follows that $\phi$ is surjective and therefore $\phi^{-1}(C_H)$ is defined.
  Let $v \in \phi^{-1}(C_H)$.
  Then $v \notin D$ and therefore $v$ is in some component $C_G$ of $G \setminus D$.
  Observation~\ref{obs:separators} implies that $\phi_R=\fproj{\phi}{D\cup C_G}$ is a homomorphism from $G[D\cup C_G]$ to $H[\phi(D)\cup C_H]$ that is locally surjective for every $v \in V(C_G)$.

  Now $|V(C_G)|<|V(C_H)|$, so there must be a vertex in $V(C_H) \setminus \phi_R(C_G)$.
  Because $C_H$ is connected, there is an edge $xy \in E(C_H)$ such that $x\in V(C_H)\setminus \phi_R(C_G)$ and $y \in \phi_R(V(C_G))$.
  But then, the mapping $\phi_z : N_G(z) \rightarrow N_H(y)$ is not surjective for any vertex $z \in \phi_R^{-1}(y)$.
\end{proof}
\fi

\iflong\subsection{Integer Linear Programming}\fi\ifshort\subparagraph*{Integer Linear Programming.}\fi

Given a set $X$ of variables and a set $C$ of linear constraints (i.e.
inequalities) over the variables in $X$ with integer coefficients, the
task in the feasibility variant of \emph{integer linear programming
  {\sc (ILP)}} is to decide whether there is an assignment $\alpha : X
\rightarrow \Z$ of the variables satisfying all constraints in $C$.
We will use the following well-known result by Lenstra~\cite{Lenstra83}.

\begin{proposition}[\cite{FellowsLokshtanovMisraRS08,FrankTardos87,Kannan87,Lenstra83}]\label{prop:Lenstra}
  \label{pro:pilp}
  {\sc ILP} is fpt parameterized by the
  number of variables.
\end{proposition}

\section{Our Algorithmic Framework}\label{s-algo}

In this section we present our main algorithmic framework that will
allow us to show that \LSHOM{}, \LBHOM{} and \textsc{Role
  Assignment} are fpt parameterized by
$k+c$, whenever the guest graph has $c$-deletion set number at most
$k$. To illustrate the main ideas behind our framework, let us first
explain these ideas for the examples of \LSHOM{} and \LBHOM{}. In this case we
are given $G$ and $H$ and we know that $G$ has
a $c$-deletion set of size at most $k$. Because of
Lemma~\ref{lem:delSet}, it then follows that
if $(G,H)$ is a yes-instance of \LSHOM{} or \LBHOM{}, then 
$H$ also has a
$c$-deletion set of size at most $k$. Informally, our next step, which
is given in Section~\ref{ssec:phds}, is to
compute a small (i.e. with size bounded by a function of $k+c$) set $\Phi$ 
of partial locally surjective homomorphisms
such that (1) every locally surjective homomorphism from $G$ to
$H$ augments some $\phi_P \in \Phi$ and (2) for every $\phi_P \in
\Phi$, the domain of $\phi_P$ is a $(k,c)$-extended deletion set of $G$
and the co-domain of $\phi_P$ is a $c'$-deletion set of $H$, where
$c'$ is bounded by a function of $k+c$.
Here and in what follows, we say that a function $\phi: V(G)\rightarrow
V(H)$ \emph{augments} (or is an \emph{augmentation} of) a partial
function $\phi_P : V_G \rightarrow V_H$, where $V_G \subseteq V(G)$
and $V_H \subseteq V(H)$ if $v\in V_G \Leftrightarrow \phi(v)\in V_H$
and $\fproj{\phi}{V_G}=\phi_P$.
This allows us to reduce our problems to
(boundedly many) subproblems of the following form: Given a
$(k,c)$-extended deletion set $D_G$ for $G$, a $c'$-deletion set
$D_H$ for $H$, and a locally surjective (respectively bijective)
homomorphism $\phi_P$ from $D_G$ to $D_H$, find a locally surjective
homomorphism $\phi$ from $G$ to $H$ that augments
$\phi_P$.
In Section~\ref{ssec:ILPform} we will then show how to formulate this subproblem as an integer linear
program and in Section~\ref{ssec:solvILP} we will show that we can efficiently
construct and solve the ILP for this subproblem.
Importantly, our ILP formulation will
allow us to solve a much more general problem, where the host graph
$H$ is not explicitly given, but defined in terms of a set of linear
constraints. We will then exploit this in Section~\ref{sec:appl} to solve not
only \LSHOM{} and \LBHOM{}, but also the \textsc{Role Assignment} problem.

\subsection{Partial Homomorphisms for the Deletion Set}
\label{ssec:phds}

For a graph $G$ and $m\in \mathbb{N}$ we let $D_G^{m}:=\{v\in
V(G)\mid \deg_G(v)\geq m\}$. The aim of this subsection is to show
that there is a small set $\Phi$ of partial homomorphisms such that
every locally surjective (respectively bijective) homomorphism from
$G$ to $H$ augments some $\phi_P \in \Phi$ and, for every $\phi_P \in
\Phi$, the domain of $\phi_P$ is a $(k,c)$-extended deletion set for~$G$ of size at most $k$
and its co-domain is a $c'$-deletion set of size at most $k$
for $H$. The main idea behind finding this set $\Phi$ is to consider
the set of high degree vertices in $G$ and $H$, i.e. the sets
$D_G^{k+c}$ and $D_H^{k+c}$. As it turns out (see
Lemma~\ref{lem:highDegreeDeletionSet}), for every subset
$D\subseteq D_G^{k+c}$, $D$ is a $(k-|D|,c)$-extended deletion set for
$G$ of size at most $k$ and $D_H^{k+c}$ is a $c'$-deletion set for~$H$ of size
at most $k$, where $c'=kc(k+c)$. Moreover, as we will show in
Lemma~\ref{lem:LBHMapsHighDegreeToHighDeptree}, every locally surjective (respectively bijective)
homomorphism from $G$ to $H$ has to augment a locally surjective
(respectively bijective) homomorphism from some induced subgraph of $G[D_G^{k+c}]$ to $D_H=D_H^{k+c}$. Intuitively, this holds because for every locally surjective homomorphism, only vertices of high degree in $G$ can be mapped to a
vertex of high degree in $H$ and for every vertex in~$H$, there must be a vertex in~$G$ that is mapped to it.
\iflong \begin{lemma} \fi \ifshort \begin{lemma}[$\star$] \fi\label{lem:highDegreeDeletionSet}
  Let $G$ be a graph. If $G$ has a $c$-deletion set of size at most $k$,
  then the set $D_G^{k+c}$ is
  a $kc(k+c)$-deletion set of size at most $k$.
  Furthermore, every subset $D\subseteq D_G^{k+c}$ is a
  $(k-|D|,c)$-extended deletion set of $G$.
\end{lemma}
\iflong
\begin{proof}
  Let $D'$ be a $c$-deletion set of $G$ of size at most $k$.
  Then every vertex $v\in V(G)\setminus D'$ has degree at most $k+c-1$, as each of its neighbours lies either in its own component of $G\setminus D'$ or in $D'$.
  Hence $D_G^{k+c}\subseteq D'$ and therefore $|D_G^{k+c}|\leq k$.
  Let $C_1,\ldots, C_m$ be the components of $G \setminus D'$ that contain a vertex adjacent to a vertex in $D'\setminus D_G^{k+c}$.
  Since $|D'\setminus D_G^{k+c}|\leq k$ and every vertex in $D'\setminus D_G^{k+c}$ has degree at most $k+c-1$, we find that $m\leq k(k+c-1)$ and $|C_1\cup \dots \cup C_m\cup (D'\setminus D_G^{k+c})|\leq k+kc(k+c-1)\leq kc(k+c)$.
  Since every component in $G\setminus D_G^{k+c}$ is either contained in a component of $G\setminus D'$ or contained in $C_1\cup \dots \cup C_m\cup (D'\setminus D_G^{k+c})$, we find that $D_G^{k+c}$ is a $kc(k+c)$-deletion set.

 Let $D\subseteq D_G^{k+c}\subseteq D'$. We will show that $D$ is a $(k-|D|,c)$-extended deletion set of $G$.
The components of $G\setminus D$ that contain no vertices from $D'\setminus D$ are components of $G\setminus D'$ and thus have size at most~$c$. Consider a component $C$ of $G\setminus D$ that contains at least one vertex from $D'\setminus D$. 
Let $D_C=V(C)\cap (D'\setminus D)$. Every component of $C\setminus D_C$ is a component of $G\setminus D'$ and thus has size at most~$c$.
Moreover, $D_C$ has size at most $|D'\setminus D|\leq k-|D|$.
 We conclude that every component of $G \setminus D$ either has size at most~$c$ or has a $c$-deletion set of size at most~$k-|D|$. Furthermore, since there are at most $k-|D|$ vertices in $D_G^{k+c} \setminus D$, and every component of $G \setminus D$ that has size larger than~$c$ must contain a vertex of $D_G^{k+c}$, it follows that there are at most $k-|D|$ components of $G \setminus D$ that have size larger than~$c$. This completes the proof.
\end{proof}
\fi
Using the above lemma, we now get to the most important result of this
subsection, which informally speaking shows that there are only boundedly many
possible pre-images of the vertices in $D_H^{k+c}$ and, moreover, these
pre-images are subsets of $D_G^{k+c}$. 

\begin{lemma}\label{lem:LBHMapsHighDegreeToHighDeptree}
  Let $G$ and $H$ be non-empty connected graphs such that $G$ has a $c$-deletion set of size at most $k$.
  If there is a locally surjective homomorphism $\phi$ from $G$ to $H$, then there is a set $D\subseteq D_G^{k+c}$ and a locally
  surjective homomorphism $\phi_P$ from $G[D]$ to $H[D_H^{k+c}]$
  such that $\phi$ augments $\phi_P$.
  If $\phi$ is locally bijective, then $D=D_G^{k+c}$ and $\phi_P$
  is a locally bijective homomorphism.
\end{lemma}

\begin{proof}	
  By Lemma~\ref{lem:highDegreeDeletionSet}, $D_G^{k+c}$ is a $kc(k+c)$-deletion set of size at most $k$.
  Furthermore, observe that for a locally surjective homomorphism $\phi$ from $G$ to $H$, the inequality $\deg_G(v) \geq \deg_H(\phi(v))$ holds for every $v\in V(G)$ ($\deg_G(v) = \deg_H(\phi(v))$ holds in the locally bijective case).
  Since $\phi$ is surjective by \cref{obs:surjective}, this implies
  that $\phi(D_G^{k+c})\supseteq D_H^{k+c}$ (and if $\phi$
  is locally bijective, then $\phi(D_G^{k+c})= D_H^{k+c}$).
  By \cref{lem:delSet}, $\phi(D_G^{k+c})$ is a $kc(k+c)$-deletion set for $H$.
  Let $D = \phi^{-1}(D_H^{k+c})$, so $D\subseteq D_G^{k+c}$ (note that $D=D_G^{k+c}$ if $\phi$ is locally bijective).
  Now $\fproj{\phi}{D}$ is a surjective map from $D$ to $D_H^{k+c}$.
  Furthermore, $\phi(D_G^{k+c}\setminus D)\cap \phi(D)=\phi(D_G^{k+c}\setminus D)\cap D_H^{k+c}=\emptyset$.
  Moreover, for every $v \in V(G) \setminus D_G^{k+c}$, $\phi(v) \notin D_H^{k+c}=\fproj{\phi}{D}(D)$, since $\deg_G(v) \geq \deg_H(\phi(v))$.
  Furthermore, $\fproj{\phi}{D}$ is a homomorphism from $G[D]$ to $H[D_H^{k+c}]$ because $\phi$ is a homomorphism.
  We argue that $\fproj{\phi}{D}$ is locally surjective  (bijective resp.) by contradiction.
  Suppose $\fproj{\phi}{D}$ is not locally surjective.
  Then there is a vertex $u\in D$ and a neighbour $v \in D_H^{k+c}$ of $\fproj{\phi}{D}(u)$ such that $v\notin \fproj{\phi}{D}(N_G(u)\cap D)$.
  Since $\phi$ is locally surjective, there must be $w\in N_G(u)\setminus D$ such that $\phi(w)=v$.
  This contradicts the fact that $\phi(V(G)\setminus D) \cap D_H^{k+c}=\emptyset$.
  Hence $\fproj{\phi}{D}$ is a locally surjective homomorphism.
  In the bijective case we just need to additionally observe that $\fproj{\phi}{D}$ restricted to the neighbourhood of any vertex $v\in D$ must be injective.
  This completes the proof.
\end{proof}

Finally, we show that we can easily compute all possible
pre-images of $D_H^{k+c}$ in any locally surjective (respectively
bijective) homomorphism from $G$ to $H$.
\iflong \begin{lemma} \fi \ifshort \begin{lemma}[$\star$] \fi\label{lem:runTimePartialHom}
  Let $G$ and $H$ be non-empty connected graphs and let $k,c$ be
  non-negative integers.
  For any $D\subseteq D_G^{k+c}$, we can compute the set $\Phi_D$
  of all locally surjective (respectively bijective) homomorphisms $\phi_P$
  from $G[D]$ to $H[D_H^{k+c}]$ in $\bigoh(|D|^{|D|+2})$ time.
  Furthermore, $|\Phi_D|\leq |D|^{|D|}$.
\end{lemma}
\iflong
\begin{proof}
Let $D \subseteq D_G^{k+c}$ and suppose there is a surjective map $\phi_P:D\rightarrow D_H^{k+c}$.
Then for every vertex $v \in D_H^{k+c}$, there must be a vertex $x \in D$ such that $\phi_P(x)=v$.
Therefore $|D_H^{k+c}|\leq |D|$, so if this condition fails, then we can immediately return that $\Phi_D=\emptyset$.

Otherwise, for each vertex of $|D|$, there are $|D_H^{k+c}|\leq |D|$ possible choices for where a map $\phi_P:D\rightarrow D_H^{k+c}$ could map this vertex.
We can list all of the at most $|D|^{|D|}$ resulting maps in $\bigoh(|D|^{|D|})$ time, and for each such map, we can check whether it is a locally surjective (respectively bijective) homomorphism in $\bigoh(|D|^2)$ time.
\end{proof}
\fi
\subsection{ILP Formulation}
\label{ssec:ILPform}

\begin{figure}
    \centering
    \includegraphics{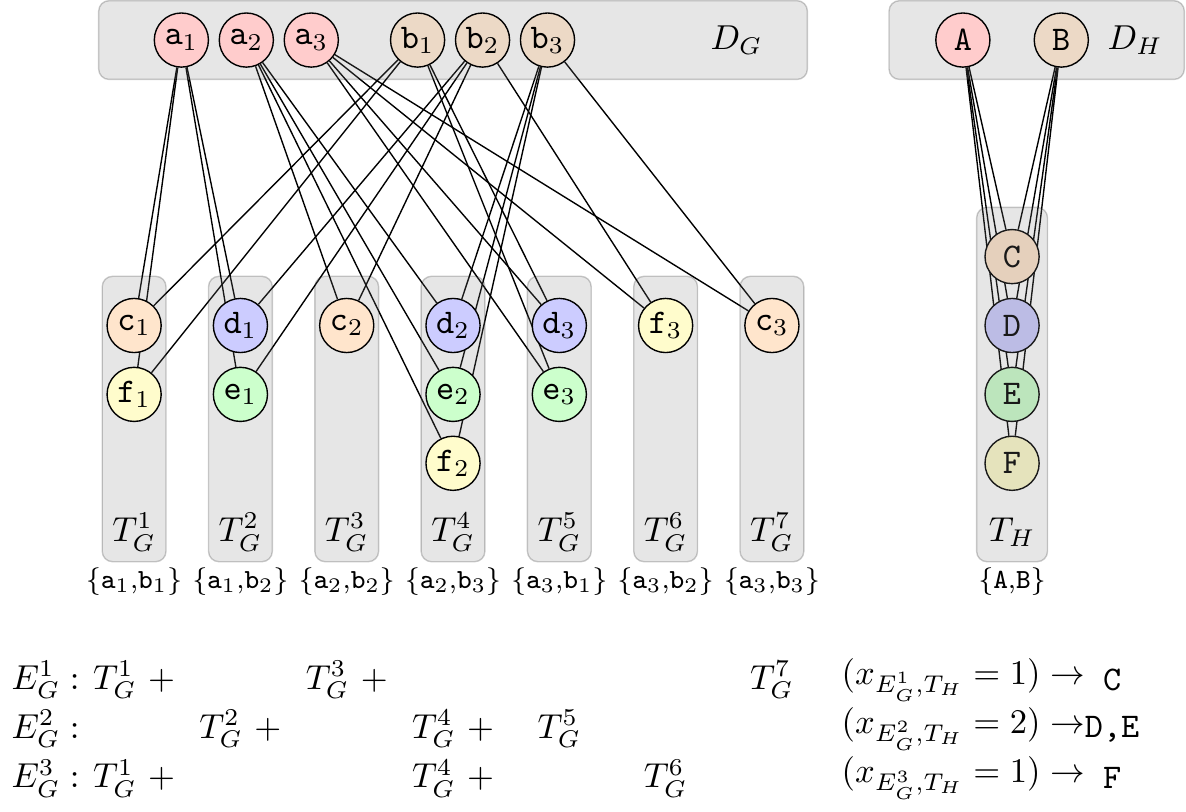}
    \bigskip
    
    \includegraphics{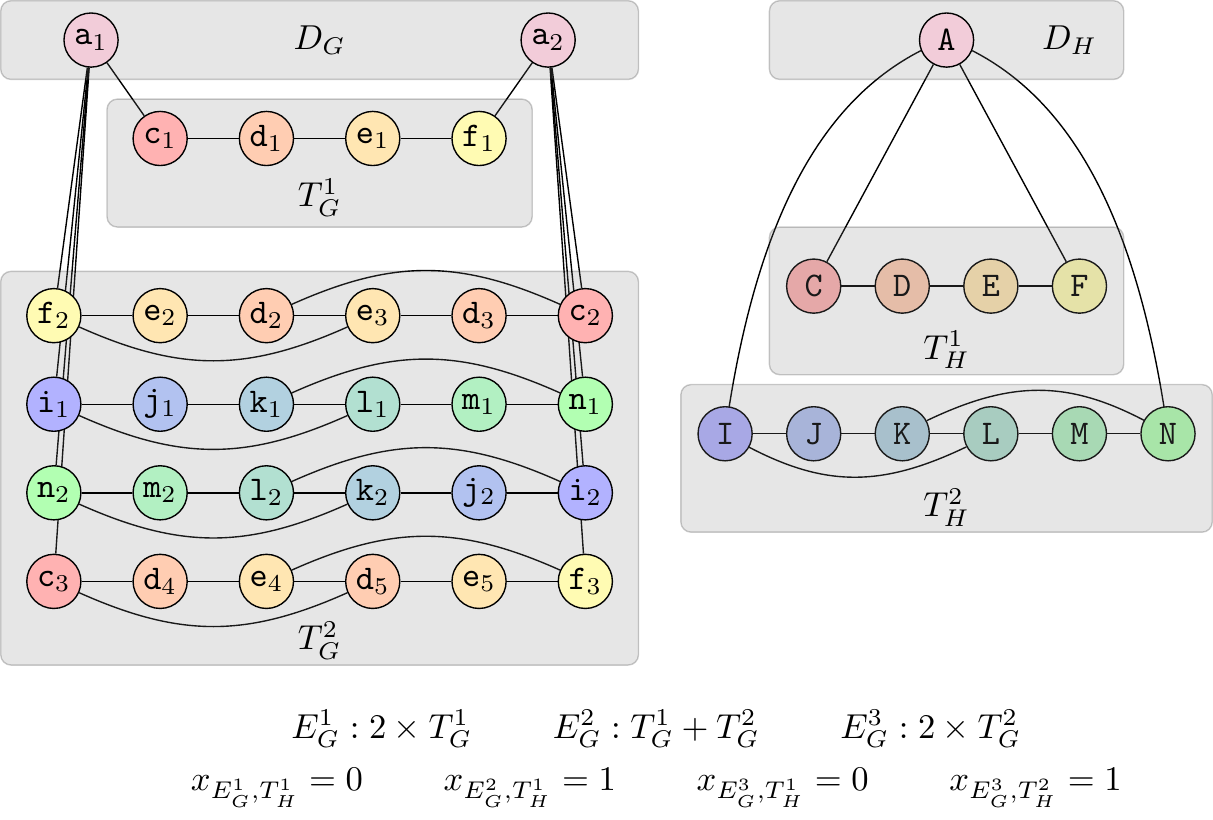}
    
    \caption{Top: A locally bijective homomorphism from a graph $G$ (left) to a graph $H$ (right), augmenting a partial homomorphism mapping the vertices of the vertex cover $D_G$ into $D_H=\{${\texttt A},{\texttt B}$\}$. 
    The $i$th vertex of $G$ mapped to some vertex {\texttt X} of $H$ is denoted {\texttt x}$_i$. Vertices in $G\setminus D_G$ are grouped by type (e.g. $\{c_1\}$ and $\{f_1\}$ have type $T_G^1$), each $T_G^i$ is characterised by the neighbours of its vertices, recalled below each column. Vertices in $H\setminus D_H$ all have the same type $T_H$.
    Rows $\EXT_G^i$ are extensions that can be minimally $\phi_P$-B-mapped to $T_H$ (in particular, each \texttt{a}$_i$ and \texttt{b}$_i$ must have a neighbour in some type in $\EXT_g^i$, which can be used as a pre-image of any vertex in $\{${\texttt C},{\texttt D},{\texttt E},{\texttt F}$\}$. Using $\EXT_G^1$ once (for colour {\texttt C}), $\EXT_G^2$ twice (for colours {\texttt D} and {\texttt E}) and $\EXT_G^3$ once (for colour {\texttt F}) yields the given locally bijective homomorphism, and it can be verified that each {\texttt a}$_i$ and {\texttt b}$_i$ indeed has all four colours in its neighbourhood.\\
    Bottom: a locally surjective homomorphism from a graph $G$ (left) to a graph $H$ (right), where $D_G$ is a $6$-deletion set. The extensions $\EXT_G^1,\EXT_G^2,\EXT_G^3$ can be minimally $\phi_P$-S-mapped to $T_H^1$; only $\EXT_G^3$ can also be minimally $\phi_P$-S-mapped to $T_H^2$. Furthermore, $T_G^1$ and $T_G^2$ can each be weakly $\phi_P$-S-mapped to some type in $H$ (respectively $T_H^1$ and $T_H^2$). Using pair $(\EXT_G^2, T_H^1)$ and $(\EXT_G^3, T_H^2)$ once is sufficient to ensure that the mapping is locally surjective for each \texttt{a}$_i$.}
    \label{fig:LBHOM}
\end{figure}

\newcommand{\noTc}[2]{(#1+#2)2^{\binom{#1+#2}{2}}} 

In this section, we will show how to formulate the subproblem obtained
in the previous subsection in terms of an ILP instance.
More specifically, we will show that the following problem can be formulated in terms of an ILP: given a partial locally surjective (respectively bijective) homomorphism $\phi_P$ from some induced subgraph $D_G$ of $G$ to some induced subgraph $D_H$ of $H$, can this be augmented to a locally surjective (respectively bijective) homomorphism from $G$ to $H$?
See the top of Figure~\ref{fig:LBHOM}
for an illustration of the subproblem for the simpler case when $D_G$ is a vertex cover of $G$ and
we are looking for a locally bijective homomorphism.
Moreover, we will actually show that for this to work, the host graph
$H$ does not need to be given explicitly, but can instead be defined by a
certain system of linear constraints.

The main ideas behind our translation to ILP are as follows. Suppose that there is a locally surjective (respectively bijective)
homomorphism $\phi$ from $G$ to $H$ that augments $\phi_P$. Because $\phi$
augments $\phi_P$, Observation~\ref{obs:separators} implies that
$\phi$ maps every component $C_G$ of $G\setminus V(D_G)$ entirely to some
component $C_H$ of $H\setminus V(D_H)$, moreover, $\fproj{\phi}{V(D_G)\cup
  V(C_G)}$ is already locally surjective (respectively bijective) for every
vertex $v \in V(C_G)$. Our aim now is to describe $\phi$ in terms of its
parts consisting of locally surjective (respectively bijective)
homomorphisms from \emph{extensions} of $D_G$ in $G$, i.e. sets of
components of $G\setminus D_G$ plus $D_G$, to \emph{simple extensions} of
$D_H$ in $H$, i.e. single components of $H\setminus D_H$ plus
$D_H$. Note that the main difficulty comes from the fact that we need
to ensure that $\phi$ is locally surjective (respectively bijective)
for every $d \in D_G$ and not only for the vertices within the
components of $G\setminus D_G$. This is why we need to describe the parts of
$\phi$ using sets of components of $G\setminus D_G$ and not just
single components. However, as we will show, it will suffice to
consider only minimal extensions of $D_G$ in $G$, where an extension is minimal if
no subset of it allows for a locally surjective (respectively
bijective) homomorphism from it to some simple extension of $D_H$ in
$H$. The fact that we only need to consider minimal extensions is
important for showing that we can compute the set of all possible
parts of $\phi$ efficiently (see Section~\ref{ssec:solvILP}).
Having shown this, we can create an ILP that has one variable
$x_{\EXT_G\EXT_H}$ for every minimal extension $\EXT_G$ and every simple
extension $\EXT_H$ such that there is a
locally surjective (respectively bijective) homomorphism from $\EXT_G$ to
$\EXT_H$ that augments $\phi_P$. The value of the variable $x_{\EXT_G\EXT_H}$
now corresponds to the number of parts used by $\phi$ that map minimal
extensions isomorphic to $\EXT_G$ to simple extensions isomorphic to
$\EXT_H$ that augment $\phi_P$. We can then use linear constraints on these variables to ensure
that:
\begin{description}
\item[(SB2')] $H$ contains exactly the right number of extensions isomorphic to
  $\EXT_H$ required by the assignment for $x_{\EXT_G\EXT_H}$,
\item[(B1')] $G$ contains exactly the right number of minimal extensions isomorphic to $\EXT_G$
  required by the assignment for $x_{\EXT_G\EXT_H}$
  (if $\phi$ is locally bijective),
\item[(S1')] $G$ contains at least the number of minimal extensions isomorphic to $\EXT_G$
  required by the assignment for $x_{\EXT_G\EXT_H}$
  (if $\phi$ is locally surjective),
\item[(S3')] for every simple extension $\EXT_G$ of $G$ that is not yet used in any part of
  $\phi$, there is a homomorphism from $\EXT_G$ to some simple
  extension of $D_H$ in $H$ that augments $\phi_P$ and is locally surjective for every
  vertex in $\EXT_G\setminus D_G$ (if $\phi$ is locally surjective).
\end{description}
Together, these constraints ensure that there is a locally surjective
(respectively bijective) homomorphism $\phi$ from $G$ to $H$ that
augments $\phi_P$. See also the bottom of Figure~\ref{fig:LBHOM} for an
illustration of the main ideas. We are now ready to formalise these
ideas. To do so, we need the following additional notation.

Given a graph $D$, an \emph{extension} for $D$ is a
graph $E$ containing $D$ as an induced subgraph. It is  \emph{simple}
if $E\setminus D$ is connected, and
\emph{complex} in general. Given two extensions $\EXT_1, \EXT_2$ of $D$, we
write  $\EXT_1 \sim_D \EXT_2$ if there is an isomorphism $\tau$ from
$\EXT_1$ to $\EXT_2$ with $\tau(d)=d$ for every $d\in D$. Then $\sim_D$ is
an equivalence relation. Let the \emph{types} of $D$, denoted
$\TTT_D$, be the set of equivalence classes of $\sim_D$
of simple extensions of $D$. We write $\TTT_D^c$ to denote the set of
types of $D$ of size at most $|D|+c$, so   $|\TTT_{D}^c|\leq \noTc{|D|}{c}$. 

Given a complex extension $E$ of $D$, let $C$ be a connected component
of $E\setminus D$. Then $C$ has type $T\in \TTT_D$ if $E[D\cup C]
\sim_D T$ (depending on the context, we also say that the extension
$E[D\cup C]$ has type $T$).
The \emph{type-count} of $E$ is the function $\ccmap_E: \TTT_D
\rightarrow \N$ 
such that $\ccmap_E(T)$ for $T \in \TTT_D$ is
the number of connected components of $E\setminus D$ with type $T$
(in particular if $E$ is simple, the type-count is $1$ for $E$ and $0$
for other types).  Note that two extensions are equivalent if and only
if they have the same type-counts; this then also implies that there is an isomorphism $\tau$ between the two extensions satisfying $\tau(d)=d$ for every $d\in D$.
We write $E\preceq E'$ if
$\ccmap_E(T)\leq \ccmap_{E'}(T)$ for all types $T\in \TTT_{D}$. If $E$
is an extension of $D$, we write $\TTT_D(E)=\{T\in \TTT_D \mid
\ccmap_E(T)\geq 1\}$ for the \emph{set of types of $E$} and $\EEE_D(E)$ for the set of simple  extensions
of $E$. Moreover, for $T \in \TTT_D$, we write
$\EEE_D(E,T)$ for the set of simple extensions in $E$ having type $T$.

A \emph{target description} is a tuple $(D_H, c, \CH)$ where $D_H$ is
a graph, $c$ is an integer and \CH is a set of linear constraints over
variables $x_T$, $T\in \TTT^c_{D_H}$. Note that a
type-count for ${D_H}$ is an integer assignment of the variables
$x_T$. A graph $H$ satisfies the target description $(D_H,c, \CH)$ if
it is an extension of $D_H$, $\ccmap_H(T)=0$ for $T
\notin  \TTT_{D_H}^c$, and setting $x_T=\ccmap_H(T)$ for all $T\in
\TTT_{D_H}^c$ satisfies all constraints in \CH.

In what follows, we assume that the following are given: the graphs
$D_G$, $D_H$, an extension $G$ of $D_G$, a target description $\TD=(D_H,
c, \CH)$, and a locally surjective (respectively bijective)
homomorphism $\phi_P:D_G\rightarrow D_H$.
Let $\EXT_G$ be an extension of $D_G$ with $\EXT_G \preceq
G$ and let $T_H \in \TTT_{D_H}^c$; note that we only consider $T_H \in
\TTT_{D_H}^c$, because we assume that $T_H$ is a type of a simple
extension of a graph $H$ that satisfies the target description $\TD$.
We say $\EXT_G$ can be
\emph{weakly $\phi_P$-S-mapped} to a
type $T_H$ if there exists an augmentation  $\phi:\EXT_G\rightarrow T_H$
of $\phi_P$ such that
$\phi$ is locally surjective for every $v\in \EXT_G\setminus D_G$.
We say that $\EXT_G$ can be
\emph{$\phi_P$-S-mapped} (respectively \emph{$\phi_P$-B-mapped}) to a
type $T_H$
if there exists an augmentation  $\phi:\EXT_G\rightarrow T_H$ of $\phi_P$
such that
$\phi$ is locally surjective (respectively locally bijective).
Furthermore, $\EXT_G$ can be \emph{minimally $\phi_P$-S-mapped} (respectively
\emph{minimally $\phi_P$-B-mapped}) to $T_H$ if $\EXT_G$ can be
$\phi_P$-S-mapped (respectively $\phi_P$-B-mapped) to $T_H$ and no other extension $\EXT_G'$
with $\EXT_G'\preceq \EXT_G$ can be $\phi_P$-S-mapped (respectively $\phi_P$-B-mapped) to $T_H$.
Let $\wSCMP{G}{D_G}{\TD}{\phi_P}$  be the set of all pairs $(T_G, T_H)$ such that $T_G\in
\TTT_{D_G}(G)$ can be weakly $\phi_P$-S-mapped to $T_H$.
Let  $\SCMP{G}{D_G}{\TD}{\phi_P}$ be the set of all pairs $(\EXT_G, T_H)$
with $\EXT_G \preceq G$, $T_H \in \TTT_{D_H}^{c}$ such that $\EXT_G$ can be
minimally  $\phi_P$-S-mapped to $T_H$ and let $\BCMP{G}{D_G}{\TD}{\phi_P}$ be the set of
all pairs $(\EXT_G, T_H)$ with $\EXT_G \preceq G$, $T_H \in \TTT_{D_H}^{c}$ 
such that $\EXT_G$ can be minimally $\phi_P$-B-mapped to $T_H$.
See also the bottom of Figure~\ref{fig:LBHOM} for an illustration of these notions.

We now build a set of linear constraints. To this end, besides
variables $x_T$ for $T\in T_H$, we introduce variables $x_{\EXT_G T_H}$
for each $(\EXT_G,T_H)\in \SCMPn$ (respectively
$\BCMPn$), where here and in what follows
$\wSCMPn=\wSCMP{G}{D_G}{\TD}{\phi_P}$, $\SCMPn=\SCMP{G}{D_G}{\TD}{\phi_P}$ and $\BCMPn=\BCMP{G}{D_G}{\TD}{\phi_P}$.
\begin{description}
\item[(S1)] $\sum_{(\EXT_G,T_H) \in
    \SCMPn}\ccmap_{\EXT_G}(T_G)*x_{\EXT_G T_H} \leq \ccmap_G(T_G)$ for every
  $T_G \in \TTT_{D_G}(G)$,
\item[(B1)] $\sum_{(\EXT_G,T_H) \in
    \BCMPn}\ccmap_{\EXT_G}(T_G)*x_{\EXT_G T_H}  =  \ccmap_G(T_G)$ for every
  $T_G \in \TTT_{D_G}(G)$,
\item[(S2)] $\sum_{\EXT_G : (\EXT_G,T_H)\in \SCMPn}x_{\EXT_G,T_H} = x_{T_H}$ for every
  $T_H \in \TTT_{D_H}$,
\item[(B2)] $\sum_{\EXT_G : (\EXT_G,T_H)\in \BCMPn}x_{\EXT_G,T_H} = x_{T_H}$ for every
  $T_H \in \TTT_{D_H}$,  
\item[(S3)]  $\sum_{(T_G,T_H) \in
    \wSCMPn} x_{T_H} \geq 1$ for every
  $T_G \in \TTT_{D_G}(G)$.
\end{description}
We refer to \iflong the bottom of \fi Figure~\ref{fig:LBHOM} for an illustration
and note that (S1) corresponds to (S1'), (B1) corresponds to (B1'), (S2) and (B2) correspond to (SB2'), and (S3) corresponds to (S3').

\iflong \begin{lemma} \fi \ifshort \begin{lemma}[$\star$] \fi\label{lem:ILPFormulation}
  Let $D_G$ and $D_H$ be graphs, let $G$ be an extension of $D_G$ and
  let $\TD=(D_H, c,  \CH)$ be a target description.
  Moreover, let $\phi_P:V(D_G)\rightarrow V(D_H)$ be a locally surjective
  (respectively bijective) homomorphism from $D_G$ to $D_H$.
  There exists a graph $H$ satisfying $\TD$ and a locally
  surjective (respectively bijective) homomorphism $\phi$ augmenting $\phi_P$
  if and only if the equation system \emph{(\CH, S1, S2, S3)} (respectively \emph{(\CH, B1,
  B2))} admits a solution.
\end{lemma}
\iflong
\begin{proof}
  Towards showing the forward direction of the claim, let $H$ be a
  graph satisfying $\TD=(D_H,c, \CH)$ and let $\phi$
  be a locally surjective (respectively bijective) homomorphism that augments $\phi_P$. 

  Consider $T_H \in \TTT_{D_H}(H)$ and let $\EXT_H \in \EEE_{D_H}(H,T_H)$. Let
  $P=G[\phi^{-1}(\EXT_H)]$; note that $D_G \subseteq V(P)$ and therefore
  $P$ is a (possibly) complex extension of $D_G$. Then because
  of Observation~\ref{obs:preimage}, we obtain that
  $\phi_R=\fproj{\phi}{P}$ is a locally surjective (respectively bijective) homomorphism that augments
  $\phi_P$ from $P$ to $\EXT_H$. Moreover,
  because of Observation~\ref{obs:separators}, it follows that
  $P\setminus D_G$ is the union of a set $W_\phi(\EXT_H)$ of components of
  $G-D_G$. Therefore, $W_\phi(\EXT_H)$ can be $\phi_P$-S-mapped (respectively $\phi_P$-B-mapped) to
  $T_H$. Moreover, if $W_\phi(\EXT_H)$ can be $\phi_P$-S-mapped to $T_H$, then $W_\phi(\EXT_H)$ also
  contains a subset $W_\phi^{\min}(\EXT_H)$ of components that can be
  minimally $\phi_P$-S-mapped to $T_H$.

  Let $\wSCMPn=\wSCMP{G}{D_G}{\TD}{\phi_P}$,
  $\SCMPn=\SCMP{G}{D_G}{\TD}{\phi_P}$ and
  $\BCMPn=\BCMP{G}{D_G}{\TD}{\phi_P}$.
  Let $X_T=\SB x_{T_H} \SM T_H \in \TTT_{D_H}(H) \SE$, $X_M=\SM x_{\EXT_G
    T_H} \SM (\EXT_G,T_H) \in \SCMPn \SE$ (respectively $X_M=\SB x_{\EXT_G
    T_H} \SM (\EXT_G,T_H) \in \BCMPn \SE$), and $X=X_T\cup X_M$.
  Let $\alpha : X \rightarrow \Nat$ be defined by setting:
  \begin{itemize}
  \item $\alpha(x_{T_H})=\ccmap_H(T_H)$ and
  \item $\alpha(x_{\EXT_G T_H})=|\SB \EXT_H \in \EEE_{D_H}(H,T_H) \SM
  \ccmap_{\EXT_G}=\ccmap_{W_{\phi}^{\min}(\EXT_H)}\SE|$ (in the locally surjective case)
  \item $\alpha(x_{\EXT_G T_H})=|\SB \EXT_H \in \EEE_{D_H}(H,T_H) \SM
    \ccmap_{\EXT_G}=\ccmap_{W_{\phi}(\EXT_H)}\SE|$ (in the locally bijective case)
  \end{itemize}
  We claim that the assignment $\alpha$ satisfies the equation system
  (\CH, S1, S2, S3) (respectively the equation system (\CH, B1, B2)). 
  Because $H$ satisfies $\TD$, it follows
  that $\alpha$ satisfies \CH.

  We start by showing the claim for the locally surjective case.
  Towards showing that (S3) is satisfied, consider a type $T_G \in
  T_{D_G}(G)$ and let $\EXT_G\in \EEE_{D_G}(G,T_G)$. Then, because of Observation~\ref{obs:separators}, the
  mapping $\fproj{\phi}{\EXT_G}$ maps $\EXT_G$ to some type $T_H \in
  \TTT_{D_H}(H)$ and therefore shows that $\EXT_G$ can be $\phi_P$-S-mapped to $T_H$.
    
  Towards showing (S1), let $T_G\in \TTT_{D_G}(G)$. Because of
  Observation~\ref{obs:separators}, every extension $\EXT_G \in \EEE_{D_G}(G,T_G)$
  satisfies $\phi(\EXT_G)\subseteq \EXT_H$ for some simple extension $\EXT_H$ of $D_H$.
  In other words $\EXT_G$ is contained in the
  pre-image of exactly one simple extension $\EXT_H$, showing that every
  extension $\EXT_G\in \EEE_{D_G}(G,T_G)$ is counted at most once on the left side of
  the inequality in (S1) and therefore the left side is at most $\ccmap_G(T_G)$. 

  Towards showing (S2), let $T_H \in \TTT_{D_H}$. Then, because
  $W_{\phi}^{\min}(\EXT_H)\neq W_{\phi}^{\min}(\EXT_H')$ for every two distinct $\EXT_H,\EXT_H' \in
  \EEE_{D_H}(H,T_H)$, we obtain:
  \begin{eqnarray*}
    \sum_{\EXT_G : (\EXT_G,T_H)\in \SCMPn}\alpha(x_{\EXT_G T_H}) & = &\sum_{\EXT_G : (\EXT_G,T_H)\in \SCMPn}|\SB \EXT_H \in \EEE_{D_H}(H,T_H) \SM
                                                             \ccmap_{\EXT_G} = \ccmap_{W_{\phi}^{\min}(\EXT_H)}\SE|\\
                                                       & = & \ccmap_H(T_H)\\
    & = & \alpha(T_H),
  \end{eqnarray*}
  as required.

  Finally, if $\phi$ is locally bijective, we only have to show (B1)
  and (B2), which can be shown very similarly to (S1) and (S2).
  That is, (B1) can be
  shown very similarly to (S1) by using the additional observation that due to
  the definition of $\alpha$ in terms of $W_\phi(\EXT_H)$ instead of
  $W_\phi^{\min}(\EXT_H)$, every simple extension $\EXT_G$
  also occurs in the pre-image of at least one simple extension $\EXT_H$.
  Moreover, (B2) can be
  shown in the same manner as (S2), since $W_\phi(\EXT_H)\neq W_\phi(\EXT_H')$ also holds
  for every two distinct $\EXT_H,\EXT_H' \in \EEE_{D_H}(H,T_H)$.
 
  Towards showing the reverse direction, let $\alpha : X \rightarrow
  \Nat$ be an assignment satisfying the equation system (\CH, S1, S2,
  S3) (respectively the equation system (\CH,B1,B2)). Let $H$ be the
  unique graph consisting of $D_H$ and $\alpha(T_H)$
  extensions of $D_H$ of type $T_H$ for every $T_H \in
  \TTT_{D_H}$. Then $H$
  satisfies $(D_H,c,\CH)$ and $\ccmap_H(T_H)=\alpha(x_{T_H})$.

  We now define a function
  $\phi : V(G) \rightarrow V(H)$, which will be a locally surjective
  (respectively, locally bijective)
  homomorphism that augments $\phi_P$ as follows.

  Let $\AAA$ be the multiset containing each pair $(\EXT_G,T_H) \in
  \SCMPn$ (respectively $\BCMPn$) exactly $\alpha(x_{\EXT_S T_H})$ times. Because
  of (S2) (respectively (B2)), there is a bijection $\gamma_{T_H}$
  between $\AAA_{T_H}=\SB (\EXT_G, T_H)\SM (\EXT_G,T_H)\in \AAA\SE$ and
  $\EEE_{D_H}(H,T_H)$ for
  every $T_H \in \TTT_{D_H}(H)$. Let $\gamma$ be the bijection between $\AAA$
  and the extensions $\EEE_{D_H}(H)$ given by
  $\gamma((\EXT_G,T_H))=\gamma_{T_H}((\EXT_G,T_H))$.

  Because the proof now diverges quite significantly for the locally surjective and locally bijective cases,
  we start by showing the remainder of the proof for the
  former case and then show how to adapt
  the proof in latter (easier) case.
  
  Because of (S1), there is a function $\beta$ from $\AAA$
  to the complex extensions of $G$ such that:
  \begin{itemize}
  \item $\ccmap_{\beta((\EXT_G,T_H))}=\ccmap_{\EXT_G}$ for every $(\EXT_G,T_H)
    \in \AAA$,
  \item $\beta(A)\cap \beta(A')=\emptyset$ for every
    two distinct $A$ and $A'$ in $\AAA$.
  \end{itemize}
  Let $A=(\EXT_G,T_H) \in \AAA$. We set $P_A=\beta(A)$ and $I_A=\gamma(A)$
  Because $A \in \SCMPn$, there is a locally surjective homomorphism $\phi_A$ from
  $P_A$ to $I_A$ that augments $\phi_P$.
  Let
  $\EEE_\AAA$ be the set of simple extensions $\EXT_G$ in $\EEE_{D_G}(G)$ for which there is
  an $A \in \AAA$ such that $\EXT_G$ is an induced subgraph of $\beta(A)$. Moreover, let
  $\bar{\EEE}_\AAA$ be the set of all remaining simple extensions in $\EEE_{D_G}(G)$, i.e. the
  set of all simple extensions $\EXT_G$ in $\EEE_{D_G}(G)\setminus \EEE_\AAA$.
  Consider a simple extension $\EXT_G$ in $\bar{\EEE}_\AAA$. Then, because of
  (S3), there is a $T_H \in \TTT_{D_H}(H)$ and a corresponding extension $\EXT_H \in \EEE_{D_H}(H,T_H)$ such that there
  is a homomorphism $\phi_{\EXT_G}$ from $\EXT_G$ to
  $\EXT_H$ that augments $\phi_P$, which is locally
  surjective for every $v \in V(\EXT_G-D_G)$. We are now ready to define $\phi : V(G) \rightarrow
  V(H)$. That is we set $\phi(v)$ to be equal to:
  \begin{itemize}
  \item $\phi_P(v)$ if $v \in D_G$,
  \item $\phi_{\EXT_G}(v)$ if $v \in V(\EXT_G)$ for some simple extension
    $\EXT_G\in\bar{\EEE}_\AAA$, and
  \item $\phi_A(v)$ if $v \in V(\EXT_G-D_G)$ for some $\EXT_G=\beta(A)$ and $A \in \AAA$.
  \end{itemize}
  It remains to show that $\phi$ is a locally surjective homomorphism
  from $G$ to $H$ that augments $\phi_P$. Clearly,
  $\phi$ augments $\phi_P$ by definition and because
  $\phi_{\EXT_G}$ does so too for every simple extension $\EXT_G$ in
  $\bar{\EEE}_\AAA$, as does $\phi_A$ for every $A \in \AAA$.
  Moreover, $\phi$ is also a homomorphism, because every edge $\{u,v\}
  \in E(G)$ is contained in $G[\EXT_G]$ for some simple extension $\EXT_G$
  in $\EEE_{D_G}(G)$ and $\phi$ maps $\EXT_G$ according to some homomorphism
  $\phi_{\EXT_G}$ (if $\EXT_G \in \bar{\EEE}_\AAA$ or some homomorphism
  $\phi_A$ (otherwise). For basically the same reason, i.e. because
  every $\phi_{\EXT_G}$ and every $\phi_A$ is locally surjective for
  every vertex in $V(G)\setminus D_G$, $\phi$ is locally surjective for
  every vertex $v \in V(G)\setminus D_G$. Towards showing that $\phi$ is also
  locally surjective for every $d\in D_G$, let $n_H$ be any neighbour of
  $\phi(d)$ in $H$. If $n_H \in \phi_P(D_G)$, then there is a neighbour
  $n_G$ of $d$ in $D_G$ with $\phi(n_G)=n_H$, because $\phi_P$ is
  locally surjective. If on the other hand $n_H \in V(\EXT_H-D_H)$ for
  some $\EXT_H \in T_{D_H}(H,T_H)$ with $T_H \in \TTT_{D_H}(H)$, then there is a neighbour
  $n_G$ of $d$ in $\beta(\gamma^{-1}(T_H))$ with $\phi(n_G)=n_H$,
  because $\phi$ (restricted to $\beta(\gamma^{-1}(T_H))$) is a
  locally surjective homomorphism from $\beta(\gamma^{-1}(T_H))$ to $\EXT_H$.

  This completes the proof for the locally surjective case.
  We now complete the proof for the locally bijective case.
  First note that because of (B1), the function
  $\beta$ from $\AAA$ to the complex extensions of $G$ is bijective.
  Moreover, if $A=(\EXT_G,T_H) \in \AAA$, then because $A \in
  \BCMPn$, there is a locally bijective homomorphism $\phi_A$ from
  $P_A=\beta(A)$ to $I_A=\gamma(A)$ that augments $\phi_P$.
  This now allows us to directly define $\phi : V(G) \rightarrow
  V(H)$. That is, we set $\phi(v)$ to be equal to:
  \begin{itemize}
  \item $\phi_P(v)$ if $v \in D_G$ and
  \item $\phi_A(v)$ if $v \in V(\EXT_G-D_G)$ for some $\EXT_G=\beta(A)$ and $A \in \AAA$.
  \end{itemize}
  It remains to show that $\phi$ is a locally bijective homomorphism
  from $G$ to $H$ that augments $\phi_P$. Note that we can assume that
  $\phi$ is already a locally surjective homomorphism that augments
  $\phi_P$, using the same arguments as for the 
  locally surjective case. Thus it only remains to show
  that $\phi$ is also locally injective for every $d \in D_G$. Suppose
  not, then there are two distinct neighbours $n_G$ and $n_G'$ that are
  mapped to the same neighbour $n_H$ of $\phi(d)$ in $H$. This is
  clearly not possible if both $n_G$ and $n_G'$ are in $D_G$ because
  $\phi_P$ is locally bijective on $D_G$. Moreover, this can also not
  be the case if exactly one of $n_G$ and $n_G'$ is in $D_G$, because
  then $n_H \in V(D_H)$, but because $\phi$ augments $\phi_P$, the
  other cannot be mapped to $D_H$. Therefore, we can assume that $n_G$
  and $n_G'$ are outside of $D_G$. Let $\EXT_H \in \EEE_{D_H}(H)$ be the
  simple extension containing $n_H$. Then, $n_G$ and $n_{G}'$ must by
  mapped by $\phi_{\gamma^{-1}(\EXT_H)}$, but this is not possible
  because $\phi_{\gamma^{-1}(\EXT_H)}$ is locally bijective.
\end{proof}
\fi

\subsection{Constructing and Solving the ILP}
\label{ssec:solvILP}

The main aim of this section is to show the following theorem.
\iflong\begin{theorem} \fi \ifshort \begin{theorem}[$\star$] \fi\label{the:decPart}
  Let $G$ be a graph, let $D_G$ be a $(k,c)$-extended deletion set (respectively a $c$-deletion set)
  of size at most $k$ for $G$, let $\TD=(D_H,c',\CH)$ be a target description and
  let 
  $\phi_P : D_G \rightarrow D_H$ be a locally surjective
  (respectively bijective) homomorphism from $D_G$ to $D_H$. Then,
  deciding whether there is a locally surjective (respectively
  bijective) homomorphism that augments $\phi_P$ from $G$ to any graph satisfying \CH
  is fpt parameterized by $k+c+c'$.
\end{theorem}

To prove Theorem~\ref{the:decPart}, we need to show that we can
construct and solve the ILP instance given in the previous section.
The main ingredient for the proof of Theorem~\ref{the:decPart} is
Lemma~\ref{lem:computingBCMandSCM}, which shows that we can efficiently compute the sets
$\wSCMPn$, $\SCMPn$, and $\BCMPn$. \ifshort
A crucial insight for its proof
is that if $(\EXT_G,\EXT_H) \in \SCMPn$ (or $(\EXT_G,\EXT_H) \in \BCMPn$),
then $\EXT_G$ consists of only boundedly many (in terms of some function of
the parameters) components, which will allow us to enumerate all
possibilities for $\EXT_G$ in fpt-time.
\fi
We
start by showing that the set $\TTT_{D_G}(G)$ can be computed
efficiently and has small size.
\iflong \begin{lemma} \fi \ifshort \begin{lemma}[$\star$] \fi\label{lem:computingcm}
  Let $G$ be a graph and let $D_G$ be a $(k,c)$-extended deletion set
  of size at most $k$ for $G$. Then, $\TTT_{D_G}(G)$ has size at
  most $k+\noTc{|D_G|}{c}$ and computing $\TTT_{D_G}(G)$ and
  $\ccmap_G$ is fpt parameterized by $|D_G|+k+c$.
\end{lemma}
\iflong
\begin{proof}
  Because $|\TTT_G(G)\setminus \TTT_G^c|\leq k$ and $|\TTT_G^c|\leq
  \noTc{|D_G|}{c}$, we obtain that $|\TTT_{D_G}(G)| \leq k
  +\noTc{|D_G|}{c}$. Moreover, we can compute $\TTT_{D_G}(G)$
  starting from the empty set and adding a simple extension $G[D_G\cup C]$ for some
  component $C$ of $G\setminus D_G$ if $G[D_G\cup C]$ is not
  equivalent with respect to $\sim_{D_G}$ to any element already added to
  $\TTT_{D_G}(G)$. Note that checking whether $G[D_G\cup C]\sim_D
  G[D_G\cup C']$ for two components $C$ and $C'$ of $G\setminus D_G$
  is fpt parameterized by $|D_G|+k+c$, because
  $G[D_G\cup C]$ has treewidth at most $|D_G|+k+c$ for every component
  $C$ of $G\setminus D_G$ (because of Proposition~\ref{pro:parrel})
  and graph isomorphism is fpt parameterized by
  treewidth~\cite{LokshtanovPPS17}. The same procedure can now also be
  used to compute all the non-zero entries of the function $\ccmap_G$
  (i.e. the entries where $\ccmap_G(T)\neq 0$), which provides us
  with a compact representation of $\ccmap_G$.
\end{proof}
\fi

\iflong
The following lemma is crucial for computing the sets $\SCMPn$ and
$\BCMPn$ that are required to construct the ILP instance. Informally, we
will show that if $(\EXT_G,\EXT_H) \in \SCMPn$ (or $(\EXT_G,\EXT_H) \in \BCMPn$),
then $\EXT_G$ consists of only boundedly many (in terms of some function of
the parameters) components, which will allow us to enumerate all
possibilities for $\EXT_G$ in fpt-time.
\iflong \begin{lemma} \fi \ifshort \begin{lemma}[$\star$] \fi\label{lem:msbm-size}
  Let $D_G$ and $D_H$ be graphs and let $\phi_P$ be a locally surjective (respectively locally bijective)
  homomorphism from $D_G$ to $D_H$. Moreover, 
  let $\EXT_G$ be an extension of $D_G$ that can be minimally $\phi_P$-S-mapped
  (respectively minimally $\phi_P$-B-mapped) to an extension $\EXT_H$ of $D_H$.
  Then, $\EXT_G\setminus D_G$ consists of at most $|D_G||\EXT_H\setminus D_H|$ components.
\end{lemma}
\iflong
\begin{proof}
  We first show the statement of the lemma for the case when $\phi_P$
  is locally surjective and therefore $\EXT_G$ can be minimally $\phi_P$-S-mapped
  to $\EXT_H$. Let 
  $\phi : V(\EXT_G) \rightarrow V(\EXT_H)$ be a locally surjective
  homomorphism that augments $\phi_P$ and exists because $\EXT_G$ can be
  $\phi_P$-S-mapped to $\EXT_H$. Let $\EXT_G'$ be an extension of $\EXT_G$ with
  $\EXT_G'\preceq \EXT_G$. Then, because of
  Observation~\ref{obs:separators}, it follows that $\fproj{\phi}{\EXT_G'}$
  is a homomorphism from $\EXT_G'$ to $\EXT_H$ that is locally surjective
  for every $v \in \EXT_G'\setminus D_G$. Therefore, $\fproj{\phi}{\EXT_G'}$
  is a locally surjective homomorphism from $\EXT_G'$ to $\EXT_H$ if and
  only if $\EXT_G'$ is such that $\fproj{\phi}{\EXT_G'}$ is locally
  surjective for every $d \in D_G$. That is, for every $d \in D_G$ and
  every neighbour $n_H$ of $\phi(d)$ in $\EXT_H'$, there has to exist a
  neighbour $n_G$ of $d$ in $\EXT_G'$ such that $\phi(n_G)=n_H$. Since
  this clearly holds if $n_H \in D_H$, because $\phi_P$ is a locally
  surjective homomorphism from $D_G$ to $D_H$, we can assume that the
  above only has to hold for every $d \in D_G$ and $n_H \in
  \EXT_H\setminus D_H$. Because $\phi$ is a locally surjective
  homomorphism from $\EXT_G$ to $\EXT_H$, it follows that for every $d \in
  D_G$ and every neighbour $n_H$ of $\phi(d)$ in $\EXT_H$, there is a
  component, say $C_{d,n_H}$, containing a neighbour $n_G$ of $d$ in
  $\EXT_G$ such that $\phi(n_G)=\phi(n_H)$; note that because $\phi$
  augments $\phi_P$, it follows that $n_G\notin D_G$ because $n_H \notin
  D_H$. Let $\EXT_G'$ be the extension of $D_G$ consisting of $D_G$ and
  all components $C_{d,n_H}$ for every $d \in D$ and $n_H \in
  \EXT_H\setminus D_H$ as above. Then, $\fproj{\phi}{\EXT_G'}$ is a locally
  surjective homomorphism from $\EXT_G'$ to $\EXT_H$ and since $\EXT_G$ is
  minimally $\phi_P$-S-mapped to $\EXT_H$ and $\EXT_G'\preceq \EXT_G$, it follows that
  $\EXT_G'=\EXT_G$. However, $\EXT_G'\setminus D_G$ consists of at most one component
  for every $d\in D_G$ and every $n_H \in \EXT_H \setminus D_H$ and
  therefore it consists of at most $|D_G||\EXT_H\setminus D_H|$
  components, which concludes the proof for the case when $\phi_P$ is
  locally surjective.

  It remains to show the statement of the lemma for the case when
  $\phi_P$ is locally bijective and $\EXT_G$ is minimally $\phi_P$-B-mapped to
  $\EXT_H$. Let $\phi : V(\EXT_G) \rightarrow V(\EXT_H)$ be a locally bijective
  homomorphism that augments $\phi_P$ and exists because $\EXT_G$ can be
  $\phi_P$-B-mapped to $\EXT_H$. Because $\phi$ is locally bijective, it is also
  locally surjective and therefore we can obtain the components
  $C_{d,n_H}$ of $\EXT_G\setminus D_G$ for $d \in D_H$ and $n_H \in
  \EXT_H\setminus D_H$ using the same arguments as in the case when
  $\phi$ was locally surjective. As before, let $\EXT_G'$ be the
  extension of $D_G$ containing all components $C_{d,n_H}$. Then, as
  we showed above, $\fproj{\phi}{\EXT_G'}$ is a locally surjective
  homomorphism from $\EXT_G'$ to $\EXT_H$. Moreover, $\fproj{\phi}{\EXT_G'}$ is
  also locally injective, because so is $\phi$. Therefore,
  $\fproj{\phi}{\EXT_G'}$ is a locally bijective homomorphism from $\EXT_G'$
  to $\EXT_H$, which because $\EXT_G$ can be minimally $\phi_P$-B-mapped to $\EXT_H$
  implies that $\EXT_G=\EXT_G'$, which concludes the proof of the lemma, because $\EXT_G'$ consists of at most
  $|D_G||\EXT_H\setminus D_H|$ components.
\end{proof}
\fi

The following proposition is a slight generalisation of \cite[Theorem
4]{CFHPT15} and will allow us to efficiently decide whether an
extension $\EXT_G$ can be (weakly) S-mapped (respectively B-mapped) to
some extension~$\EXT_H$.
\iflong\begin{lemma}[{\cite[Theorem 4]{CFHPT15}}]\fi\ifshort\begin{lemma}[$\star$]\fi\label{lem:homtw}
  Let $G$ and $H$ be graphs and let $\phi_P : D_G \rightarrow D_H$ be
  a locally surjective (respectively bijective) homomorphism
  from $D_G$ to $D_H$ for some subgraphs $D_G$ of $G$ and $D_H$ of
  $H$. Then deciding whether there is a locally surjective
  (respectively bijective) homomorphism from $G$ to $H$ that augments
  $\phi_P$ can be achieved in 
  $\bigoh(|V(G)|((|V(H)2^{\Delta(H)})^{\omega})^2\omega\Delta(H))$ time and
  is therefore fpt parameterized by
  $\omega+|V(H)|$, where $\omega$ is the treewidth of $G$.
\end{lemma}
\iflong
\begin{proof}
  In~\cite[Theorem 4]{CFHPT15}, the authors provide an algorithm
  that, given a graph $G$ and a graph $H$, decides in
  $\bigoh(|V(G)|((|V(H)2^{\Delta(H)})^{\omega})^2\omega\Delta(H))$ time
  whether there is a locally surjective homomorphism from $G$ to $H$,
  where $\omega$ is the treewidth of $G$. The algorithm uses a standard
  dynamic programming approach on a tree decomposition of $G$ (of width
  $\omega$), and it is straightforward to verify that the algorithm can be adapted
  with only minor modifications to an algorithm using the same
  run-time that decides whether there is a locally bijective homomorphism
  from $G$ to $H$. Similarly, it is straightforward to adapt their
  algorithm to the case that one is additionally given a locally
  surjective (respectively bijective) homomorphism $\phi_P$ from some induced subgraph
  $D_G$ of $G$ to some induced subgraph $D_H$ of $H$ and one only looks
  for a locally surjective (respectively bijective) homomorphism from
  $G$ to $H$ that augments $\phi_P$. 
\end{proof}
\fi
The following corollary now follows directly from
Lemma~\ref{lem:homtw} and the definition of (weakly) S-mapped
(respectively B-mapped).
\begin{corollary}\label{cor:testmapps}
  Let $D_G$ and $D_H$ be graphs and let $\phi_P$ be
  a locally surjective (respectively bijective) homomorphism
  from $D_G$ to $D_H$. Let $\EXT_G$ be an extension of $D_G$ having
  treewidth at most $\omega$ and let $\EXT_H$ be an extension of
  $D_H$. Then, testing whether $\EXT_G$ can be weakly $\phi_P$-S-mapped,
  $\phi_P$-S-mapped, or $\phi_P$-B-mapped to $\EXT_H$ is fpt 
  parameterized by $\omega+|\EXT_H|$.
\end{corollary}

We are now ready to show that we can efficiently compute the sets
$\wSCMPn$, $\SCMPn$, and $\BCMPn$, which is the last crucial step towards
constructing the ILP instance.
\fi
\iflong \begin{lemma} \fi \ifshort \begin{lemma}[$\star$] \fi\label{lem:computingBCMandSCM}
  Let $G$ be a graph, let $D_G$ be a $(k,c)$-extended deletion set (respectively a $c$-deletion set)
  of size at most $k$ for $G$, let $\TD=(D_H,c',\CH)$ be a target description and
  let $\phi_P$ be a locally surjective
  (respectively bijective) homomorphism from $D_G$ to $D_H$. Then,
  the sets $\wSCMPn=\wSCMP{G}{D_G}{\TD}{\phi_P}$ and
  $\SCMPn=\SCMP{G}{D_G}{\TD}{\phi_P}$ (respectively the set $\BCMPn=\BCMP{G}{D_G}{\TD}{\phi_P}$) can be computed in fpt-time
  parameterized by $k+c+c'$ and $|\SCMPn|$ (respectively $|\BCMPn|$)
  is bounded by a function depending only on $k+c+c'$.
  Moreover, the number of variables in the equation system (CH, S1, S2,
  S3) (respectively (CH, B1, B2))
  is bounded by a function depending only on $k+c+c'$.
\end{lemma}
\iflong
\begin{proof}
  We only show the lemma for the set $\SCMPn$, since the proof for the
  set $\wSCMPn$ can be seen as a special case and the proof for the
  set $\BCMPn$ is identical.
  Let $(\EXT_G,T_H)\in
  \SCMPn$. Then, $\EXT_G$ is an extension of $D_G$ with $\EXT_G\preceq G$,
  $T_H \in \TTT_{D_H}^{c'}$, and $\EXT_G$ can be minimally $\phi_P$-S-mapped to
  $T_H$. Because $\EXT_G$ can be minimally $\phi_P$-S-mapped to $\EXT_H$,
  Lemma~\ref{lem:msbm-size} implies that $\EXT_G\setminus D_G$ consists of
  at most $\ell=|D_G||\EXT_H\setminus D_H|$ components and, because $\EXT_G \preceq G$,
  these are also components of $G\setminus D_G$. Therefore, there are
  at most $(|\TTT_{D_G}(G)|)^{\ell}$ non-isomorphic possibilities for
  $\EXT_G$, which together with Lemma~\ref{lem:computingcm} and the facts
  that $\ell\leq kc'$ and $|\TTT_{D_H}^{c'}|\leq \noTc{k}{c'}$ shows
  that $|\SCMPn|\leq (|\TTT_{D_G}(G)|)^{\ell}|\TTT_{D_H}^{c'}|\leq
  (k+\noTc{k}{c})^\ell(\noTc{k}{c'})$. Therefore, $|\SCMPn|$ is
  bounded by a function depending only on $k+c+c'$. Towards showing
  that we can compute $\SCMPn$ is fpt-time parameterized by $k+c+c'$,
  first note that the set $\TTT_{D_G}(G)$ can be computed in fpt-time
  parameterized by $k+c$ using
  Lemma~\ref{lem:computingcm}. Similarly, the set $\TTT_{D_H}^{c'}$ can
  be computed in fpt-time parameterized by $k+c'$ using the same idea
  as in Lemma~\ref{lem:computingcm}. This now allows us to compute 
  the set $\AAA$ containing all non-isomorphic
  possibilities for $\EXT_G$, i.e. the set of all extensions $\EXT_G$ of $D_G$
  with $\EXT_G\preceq G$ and $\sum_{T_G \in \TTT_{D_G}(G)}\ccmap_{\EXT_G}(T_G)\leq \ell$ in fpt-time
  parameterized by $k+c+c'$, i.e. in time at most $(|\TTT_{D_G}(G)|)^{\ell}$.
  But then, $\SCMPn$ is equal to the set of all pairs $(\EXT_G,\EXT_H) \in \AAA \times
  \TTT_{D_H}^{c'}$ such that $\EXT_G$ can be minimally
  $\phi_P$-S-mapped to $\EXT_H$. Moreover, for every such pair
  $(\EXT_G,\EXT_H)$ we can test in fpt-time parameterized by $k+c+c'$
  whether $\EXT_G$ can be $\phi_P$-S-mapped to $\EXT_H$ using
  Corollary~\ref{cor:testmapps}, because the treewidth of $\EXT_G$ is at
  most $k+c$ (Proposition~\ref{pro:parrel}). Therefore, 
  we can compute $\SCMPn$ by enumerating all pairs $(\EXT_G,\EXT_H) \in \AAA \times
  \TTT_{D_H}^{c'}$, testing for each of them whether $\EXT_G$ can be
  $\phi_P$-S-mapped to $\EXT_H$ using Corollary~\ref{cor:testmapps}, and
  keeping only those pairs $(\EXT_G,\EXT_H)$ such that $\EXT_G$ can be
  $\phi_P$-S-mapped to $\EXT_H$ and $\EXT_G$ is inclusion-wise minimal among
  all pairs $(\EXT_G',\EXT_H)$.
\end{proof}
\fi
\iflong
We are now ready to prove the main result of this subsection.
\begin{proof}[Proof of Theorem~\ref{the:decPart}]
  We first compute the sets
  $\wSCMPn$ and
  $\SCMPn$ (respectively the set $\BCMPn$), which because of
  Lemma~\ref{lem:computingBCMandSCM} can be achieved in fpt-time
  parameterized by $k+c+c'$. This now allows us to
  construct the ILP instance $\III$ given by the equation system
  (\CH,S1,S2,S3) (respectively the equation system (\CH,B1,B2)) in
  fpt-time parameterized by $k+c+c'$.
  Moreover, because the number of variables in
  $\III$ is bounded by a function of $k+c+c'$ and we can employ
  Proposition~\ref{prop:Lenstra} to solve $\III$ in fpt-time
  parameterized by $k+c+c'$. Finally, because of
  Lemma~\ref{lem:ILPFormulation}, it follows that $\III$ has a solution
  if and only if there is a locally surjective   (respectively
  bijective)  homomorphism that augments $\phi_P$ from $G$
  to any graph satisfying \CH, which completes the proof of the theorem. 
\end{proof}
\fi

\section{Applications of Our Algorithmic Framework}\label{sec:appl}

In this section we show the main results of our paper, which can be obtained as an application of our framework
given in the previous section. 
Our first result implies that \LSHOM{} and \LBHOM{} are fpt parameterized by the fracture number of the guest graph.
\begin{theorem}\label{thm:FPTalgoLSHOMandLBHOM}
  \LSHOM{} and \LBHOM{} are fpt parameterized by $k+c$, where $k$ and $c$ are
  such that the guest graph~$G$ has a $c$-deletion set of size at most $k$.
\end{theorem}
\begin{proof}
  Let $G$ and $H$ be non-empty connected graphs such that $G$ has
  a $c$-deletion set of size at most $k$. Let $D_H=H[D_H^{k+c}]$.
  We first verify whether $H$
  has a $c$-deletion set of size at most $k$ using Proposition~\ref{pro:comp-dels}.
  Because of Lemma~\ref{lem:delSet}, we can return that there is no
  locally surjective (and therefore also no bijective) homomorphism from $G$ to $H$ if this is not the
  case. Therefore, we can assume in what follows that $H$ also has a
  $c$-deletion set of size at most $k$, which together with
  Lemma~\ref{lem:highDegreeDeletionSet} implies that $V(D_H)$ is a
  $kc(k+c)$-deletion set of size at most $k$ for $H$. Therefore,
  using Lemma~\ref{lem:computingcm}, we can compute $\ccmap_H$ in
  fpt-time parameterized by $k+c$. This now allows us to obtain a target
  description $\TD=(D_H,c',\CH)$ with $c'=kc(k+c)$ for $H$, i.e.
  $\TD$ is satisfied only by the graph $H$, by adding the
  constraint $x_T=\ccmap_H(T_H)$ to $\CH$ for every simple extension type $T_H
  \in \TTT^{c'}_{D_H}$; note that $\TTT^{c'}_{D_H}$ can be computed
  in fpt-time parameterized by $k+c$ because of Lemma~\ref{lem:computingcm}.

  Because of
  Lemma~\ref{lem:LBHMapsHighDegreeToHighDeptree}, we obtain that there
  is a locally surjective (respectively bijective) homomorphism $\phi$ from $G$ to $H$ if and only if
  there is a set $D \subseteq D_G^{k+c}$ and a locally surjective (respectively bijective)
  homomorphism $\phi_P$ from $D_G=G[D]$ to $D_H$
  such that $\phi$ augments $\phi_P$. Therefore, we can solve \LSHOM{}
  by checking, for every $D\subseteq D_G^{k+c}$ and every
  locally surjective homomorphism $\phi_P$ from $D_G=G[D]$ to $D_H$,
  whether there is a locally surjective homomorphism from $G$ to $H$
  that augments $\phi_P$. Note that there are at most $2^k$ subsets
  $D$ and because of Lemma~\ref{lem:runTimePartialHom}, we can compute the set
  $\Phi_D$ for every such subset in
  $\bigoh(k^{k+2})$ time. Furthermore, because of
  Lemma~\ref{lem:highDegreeDeletionSet}, $D$ is a $(k-|D|,c)$-extended
  deletion set of size at most $k$ for $G$. Therefore,
  for  every $D\subseteq D_G^{k+c}$ and $\phi_p \in \Phi_D$,
  we can employ
  Theorem~\ref{the:decPart} to decide in fpt-time parameterized by
  $k+c$ (because $c'=kc(k+c)$), whether there is a locally
  surjective (respectively bijective) homomorphism from $G$ to a graph satisfying $\TD$
  that augments $\phi_P$. Since $H$ is the only graph satisfying
  $\TD$, this completes the proof of the theorem.
\end{proof}

The proof of the following theorem is similar to the proof of
\cref{thm:FPTalgoLSHOMandLBHOM}. The major difference is that $H$ is
not given. Instead, we use \cref{the:decPart} for a selected set of
target descriptions. Each of these target descriptions enforces that
graphs satisfying it have to be connected and have precisely $h$
vertices, where $h$ is part of the input for the {\sc Role Assignment}
problem. Furthermore, we ensure that every graph~$H$ satisfying the
requirements of the {\sc Role Assignment} problem must satisfy at least one
of the selected target descriptions. The size of the set of considered
target descriptions  depends only on $c$ and $k$, as it is sufficient
to consider any small graph $D_H$ and  types of small simple
extensions of $D_H$. 
\iflong \begin{theorem} \fi \ifshort \begin{theorem}[$\star$] \fi
  {\sc Role Assignment} is fpt parameterized by $k+c$, where $k$ and $c$ are
  such that $G$ has a $c$-deletion set of size at most $k$.
\end{theorem}
\iflong
\begin{proof}
Let $G$ be a non-empty connected graph such that $G$ has a $c$-deletion set of size at most $k$ and let $h\geq 1$ be an integer. 

In order to use \cref{the:decPart} in this case, we need to ensure that the target descriptions used enforce that $H$ is connected and has $h$ vertices. 
Therefore for a fixed graph $D$ on at most $k$ vertices, we let $\CON_{D}$ be the set of all minimal sets $S\subseteq \TTT_{D}^{k+c}$ such that any extension $H$ of $D$, which contains exactly the types in $S$ is connected. Since $|\TTT_{D}^{k+c}|$ is bounded by
$(2k+c)2^{{2k+c \choose 2}}$,
 we can compute $\CON_D$ by considering every $S\subseteq \TTT_D^{k+c}$ and checking whether an extension $T\in \TTT_D$ of $D$ containing precisely the types in $S$ is connected. Since $|V(T)|\leq k+(k+c)\cdot |S|$ and checking connectivity takes linear time (using BFS or DFS) we can compute $\CON_D$ in time depending only on $k$ and $c$. For $S\in \CON_D$, we set $\CH_S$ to be the set of equations containing 
   $x_T\geq 1$ for every $T\in S$ and 
    $|V(D_H)|+\sum_{T\in \TTT_{D_H}^c} (|V(T)|-|V(D_H)|)*x_T=h$. Note that for $D$ and $S\in \CON_D$, any graph $H$ satisfying the target description $(D,c+k,\CH_S)$ is connected and has $h$ vertices.

  If there is a connected graph $H$ on $h$ vertices and a locally surjective homomorphism $\phi$ from $G$ to $H$, then by Lemma~\ref{lem:LBHMapsHighDegreeToHighDeptree} 
  there is a set $D \subseteq D_G^{k+c}$ and a locally surjective 
  homomorphism $\phi_P$ from $D_G=G[D]$ to $D_H=H[D_H^{k+c}]$
  such that $\phi$ augments $\phi_P$. Note that by Lemmas~\ref{lem:delSet} and~\ref{lem:highDegreeDeletionSet}, $D_H$ is a $(k+c)$-deletion set of size at most $k$. This implies firstly that $D_H$ is a graph on at most $k$ vertices. Secondly,  $H$ is an extension of $D_H$ such that $\ccmap_H(T)=0$ for $T\notin T_{D_H}^{c+k}$ and, since $H$ is also connected and has $h$ vertices, $H$ satisfies the target description $(D_H,c+k,\CH_S)$ for at least one $S\in \CON_{D_H}$.

  Therefore, we can solve the {\sc Role Assignment} problem 
  by checking for every $D\subseteq D_G^{k+c}$, every graph $D_H$ on no more than $k$ vertices, every $S\in \CON_{D_H}$ and every
  locally surjective homomorphism $\phi_P$ from $D_G=G[D]$ to $D_H$,
  whether there is a graph $H$ satisfying the target description $(D_H,k+c,\CH_S)$ and a  locally surjective homomorphism from $G$ to $H$
  that augments $\phi_P$. 
  Note that there are at most $2^k$ subsets
  $D$. Furthermore, there are at most
$k2^{k \choose 2}$
graphs on at most $k$ vertices and for each we can compute $\CON_D$ in time depending only on $k$ and $c$.
For each such graph $D_H$, there are at most
$|\CON_D|\leq 2^{(2k+c)2^{{2k+c \choose 2}}}$
subsets $S$ to consider. Lastly, because of Lemma~\ref{lem:runTimePartialHom}, for every $D\subseteq D_G^{k+c}$, and any graph $D_H$ on no more than $k$ vertices, we can compute the set
  of locally surjective homomorphisms $\phi_P$ from $G[D]$ to $D_H$ in time  $\bigoh(k^{k+2})$ time and there are at most $|D|^{|D|}$ $\phi_P$ to consider. 
  
  By
  Lemma~\ref{lem:highDegreeDeletionSet}, $D$ is a $(k-|D|,c)$-extended
  deletion set of size at most $k$ for $G$. Therefore,
  for  every $D\subseteq D_G^{k+c}$, every graph $D_H$ on no more than $k$ vertices, every $S\in \CON_{D_H}$ and every
  locally surjective homomorphism $\phi_P$ from $D_G=G[D]$ to $D_H$,
  we can employ
  Theorem~\ref{the:decPart} to decide in fpt-time parameterized by
  $k+c$, whether there is a graph $H$ satisfying $(D_H,c+k,\CH_S)$ and a locally
  surjective  homomorphism from $G$ to $H$
  that augments $\phi_P$. This completes the proof. 
\end{proof}
\fi

\section{Locally Injective Homomorphisms}\label{s-injective}

\iflong
The following result is well known. We include a proof for completeness.

\begin{theorem}[Folklore]\label{the:IH-wh}
  \LIHOM{} is \W{1}-hard parameterized by $|V(G)|$. In particular, it is \W{1}-hard for all structural parameters of $G$.
\end{theorem}
\begin{proof}
Let $G$ be a complete graph on $k$ vertices, and let $H$ be an arbitrary graph. There exists a locally injective homomorphism $\phi$ from $G$ to $H$ if and only if $H$ contains a clique $K$ on $k$ vertices.
Indeed, for the forward direction, pick $K$ to be the image of $V(G)$ under $\phi$.
Then $|K|=|V(G)|=k$ by the local injectivity of $\phi$, and $K$ is a clique.
For the reverse direction, let $\phi$ be any bijection between $V(G)$ and $K$.
The result follows from the fact that \textsc{Clique} is \W{1}-hard.
\end{proof}
\fi

The locally injective case is more difficult in our setting since, in general, surjectivity helps to transfer structural parameters on $G$ to similar structures on $H$ (for example, in \LSHOM{} and \LBHOM{} the image of a deletion set is also a deletion set by \Cref{lem:delSet}). In \LIHOM{} however, and even in the restricted case of graphs with bounded vertex cover number,  no such property can be used to help find the image of a vertex cover, and exponential-time enumerations appear to be necessary. On the positive side, once such a partial mapping from a vertex cover of $G$ to $H$ has been found, our ILP framework can still be applied to map the remaining vertices in \FPT-time. This leads to an \XP algorithm for vertex cover number (\Cref{the:IHvc}). Interestingly, this result does not extend to $c$-deletion set number for $c>1$: even if the mapping of the deletion set can be guessed, the fact that the non-trivial remaining components must be mapped to distinct subgraphs of $H$ makes the problem difficult (see \Cref{the:IHdic}).\iflong \begin{theorem} \fi \ifshort \begin{theorem}[$\star$] \fi\label{the:IHvc}
  \LIHOM{} is in \XP{} parameterized by the vertex cover number of $G$.
\end{theorem}
\iflong
\begin{proof}
As for the surjective and bijective cases, we employ a two-step algorithm that first guesses the image of the vertex cover through a partial homomorphism, then runs an ILP to map the remaining vertices. The ILP only requires \FPT-time, however the first step needs an exhaustive enumeration of subsets of $H$ (in the injective case, the image of a vertex cover does not have to be a vertex cover), hence the \XP{} running time.

We use the definitions of \emph{types} and \emph{extensions} from Section~\ref{ssec:ILPform}. 
Let $G$ be a connected graph with a vertex cover $D_G$ of size $k$.  Note that connected components of $G\setminus D_G$ are single vertices.
We can thus define the type of a vertex $v\in G\setminus D_G$ to be the type of the component $\{v\}$. Note that there are at most $n_t = |\TTT_{D_G}(G)|\leq 2^{|D_G|}$ types in $G$.

The first step of the algorithm consists of guessing a partial homomorphism $\phi_P$ between $D_G$ and $H$ (there are $|V(G)|^{|D_G|}$ such homomorphisms). From now on, we look for locally injective homomorphisms $\phi$ from $G$ to $H$ such that $\phi(v)\in D_H\Rightarrow v\in D_G$ (such a function $\phi$ is more simply called a \emph{solution}). 
We write $D_H=\phi_P(D_G)$. $\phi$ is \emph{stable} if if it is an augmentation of $\phi_P$ (i.e. if  $\phi(v)\in D_H\Leftrightarrow v\in D_G$).

Let $\phi$ be a  solution. For a vertex $h\in V(H)$, let $P_h=\phi^{-1}(h)$ be the pre-image of $h$. Sets $P_h$ form a partition of $V(G)$, two distinct vertices of $P_h$ may not share a neighbour in $G$ (otherwise, $\phi$ would not be locally injective for this common neighbour). Thus, $P_h\setminus D_G$ has size at most $|D_G|$ (since by connectivity of $G$, each vertex in $G\setminus D_G$ has at least one neighbour in $D_G$), and may not contain two vertices with the same type. In particular there are at most $n_t^{|D_G|}$ distinct pre-images (up to equivalence).

Guess the pre-image of each $h\in D_H$ (for a total of $n_t^{|D_G|^2}$ branches), and let $D'_G=\bigcup_{h\in D_H} P_h$. Then $D'_G$ is a vertex cover of $G$ with size at most $|D_G|^2$. Define $\phi'_P:D'_G\rightarrow D_H$ such that $\phi'_P(v) = h$ whenever $v\in P_h$. Thus $\phi$ is a locally injective homomorphism extending $\phi_P'$ and $\phi(v)\in D_H\Leftrightarrow v\in D_G'$. Without loss of generality, we thus assume that we look for a stable solution $\phi$ (equivalently, we can set $D_G:=D_G'$ and $\phi_P:=\phi_P'$). Also note that all edges in $H\setminus D_H$ can be safely ignored (since no solution $\phi$ would map an edge $\{u,v\}$ of $G$ to such an edge of $H$, and there is no surjectivity constraint), so that $D_H$ is a vertex cover of $H$.   Finally, we can assume that $\phi_P$ is locally injective (on its domain $D_G$), and that two vertices $u,u'$ with $\phi_P(u)=\phi_P(u')$ do not share a neighbour in $G$, since otherwise no stable solution exist.

A (possibly empty) subset $P$ of $V(G\setminus D_G)$ is a \emph{candidate pre-image} of $h\in V(H\setminus D_H)$ if the following conditions hold:
\begin{enumerate}
\item \label{cond:candidate_homom} $\phi_P(N_G(P)) \subseteq N_H(h)\cap D_H$,
\item \label{cond:common_neighbour} any two vertices in $P$  do not share a neighbour,
\item \label{cond:type-count} $P$ contains at most one vertex from every type in $G$ and
\item \label{cond:size} $P$ has size at most $|D_G|$.
\end{enumerate}
By the remarks above, given a stable solution $\phi$, $P_h=\phi^{-1}(h)$ is a candidate pre-image of $h$. Conversely, building $\phi$ using candidate pre-images only leads to a stable solution, as formalised below.

\begin{claim}
 If $\phi:G\rightarrow H$ satisfies $\phi(v)=\phi_P(v)$ for $v\in D_G$ and $\phi^{-1}(h)$ is a candidate pre-image of $h$ for each $h\in V(H\setminus D_H)$, then $\phi$ is a stable solution.
\end{claim}
\begin{claimproof}
First note that $\phi$ is a homomorphism, i.e. $(\phi(u),\phi(v))$ is an edge in $H$ for each edge $(u,v)$ in  $G$ (since $\phi_P$ is a homomorphism for $u,v\in D_G$,   by Condition~\ref{cond:candidate_homom} for $u\in D_G$ and $v\notin D_G$, and the case $u,v\notin D_G$ is impossible since $D_G$ is a vertex cover). 
Pick $v\in V(G)$ and $u,u'\in N_G(v)$. We prove that $\phi(u)\neq \phi(u')$.
If $u,u'\in D_G$, then $u,u'$ share a neighbour and $\phi_P(u)\neq\phi_P(u')$.
If $u\in D_G$ and $u'\notin D_G$, then $\phi(u)\in D_H$ and $\phi(u')\notin D_H$.
If $u,u'\notin D_G$, then $u,u'$ share a neighbour and they cannot be in the same candidate pre-image (by Condition~\ref{cond:common_neighbour}).
\end{claimproof}

Also note that if $D_G\cup P \sim_{D_G} D_G\cup P'$ and $D_H\cup \{h\} \sim_{D_H} D_H\cup\{h'\}$, then $P$ is a candidate pre-image of $h$ if and only if $P'$ is a candidate pre-image of $h'$. Let $\ICMP$ be the set of pairs $(\EXT_G,T_H), \EXT_G\in \EEE_{D_G}(G), T_H\in \TTT_H$ such that $\EXT_G$ contains an extension $D_G\cup P$, $T_H$ contains an extension $D_H\cup \{h\}$, and $P$ is a candidate pre-image of $h$ (note that there is no minimality constraint for pairs in $\ICMP$). 

We now build the ILP computing the pre-images of vertices in $H\setminus D_H$. Introduce a variable $x_{\EXT_G,T_H}$ for each pair $(\EXT_G,T_H)\in\ICMP$. This variable represents the number of vertices $h$ with type $T_h$ whose pre-image $P_h$ has $P_h\cup D_G \in \EXT_G$. We introduce two types of constraints (see below). Constraint (I1) enforces that the pre-images $P_h$ form a partition of $V(G\setminus D_G)$ by counting vertices of each type in each $\EXT_G$ and checking that the sum corresponds to the count in $G$. Constraint (I2) corresponds to the fact that each vertex in $H$ needs to be assigned a (possibly empty) pre-image (the number of pairs involving a type $T_H$ must correspond to the type-count of  $T_H$ in $H$).
\begin{description}
	\item[(I1)] $\sum_{(\EXT_G,T_H) \in
	\ICMP}\ccmap_{\EXT_G}(T_G)*x_{\EXT_G, T_H}  =  \ccmap_G(T_G)$ for every
$T_G \in \TTT_{D_G}(G)$,
	\item[(I2)] $\sum_{\EXT_G : (\EXT_G,T_H)\in \ICMP}x_{\EXT_G,T_H} = \ccmap_H(T_H)$ for every
	$T_H \in \TTT_{D_H}$.	
\end{description}

From the above remarks, a stable solution $\phi$ yields a feasible solution for constraints (I1,I2). Conversely, a solution to the ILP gives integers $x_{\EXT_G,T_H}$: for each pair $(\EXT_G,T_H)$, pick $x=x_{\EXT_G,T_H}$ new vertices $h_1,\ldots,h_x$ with type $T_H$, pick  $x$ sets $P_1,\ldots,P_x$ in $G\setminus D_G$ such that $P_i\cup D_G \in \EXT_G$ (each $P_i$ being disjoint from previously selected sets). Assign $\phi(v)=h_i$ for each $v\in P_i$, and $\phi(v)=\phi_P(v)$ for $v\in D_G$. Then, over all pairs $(\EXT_G,T_H)$, the sets $P_i$ form a partition of $V(G\setminus D_G)$ so $\phi$ is well defined, and each $\phi^{-1}(h)$ is a candidate pre-image of $H$. By the claim above, $\phi$ is indeed a locally injective homomorphism from $G$ to~$H$.
\end{proof}

\begin{corollary}\label{cor:lihom-a-1-p}
 For any constant $k$, \LIHOM{} is polynomial-time solvable for graphs $G$ with $1$-deletion set number at most~$k$.
\end{corollary}
\fi

We actually obtain the following dichotomy for the complexity of \LIHOM{}, where the $c=1$, $k\geq 1$ case is
already given by
\iflong
Corollary~\ref{cor:lihom-a-1-p}.
\else
Theorem~\ref{the:IHvc}.
\fi
\iflong \begin{theorem} \fi \ifshort \begin{theorem}[$\star$] \fi\label{the:IHdic}
  Let $c,k \geq 1$. Then
  \LIHOM{} is polynomial-time solvable on guest graphs with a $c$-deletion set of size at most $k$ if
  either $c= 1$ and $k\geq 1$ or $c=2$ and $k=1$; otherwise, it is \NP{}\hy complete.
\end{theorem}
\iflong
Theorem~\ref{the:IHdic} follows from Corollary~\ref{cor:lihom-a-1-p} and the following 
three 
lemmas.
\iflong \begin{lemma} \fi \ifshort \begin{lemma}[$\star$] \fi
  \LIHOM{} is polynomial-time solvable for graphs $G$ with a $2$-deletion set number at most~$1$.
\end{lemma}
\begin{proof}
Let~$G$ and~$H$ be connected graphs such that~$G$ has $2$-deletion set number at most~$1$.
If~$G$ has a $2$-deletion set containing no vertices, then~$G$ contains at most two vertices, in which case we can solve \LIHOM{} in polynomial time.
Otherwise, we can find a $2$-deletion set~$\{v\}$ in polynomial time by trying all possibilities for~$v$.
Let~$p$ be the number of edges in $G[N_G(v)]$ and let~$w$ be a vertex of~$H$.
We claim that there is a locally injective homomorphism~$\phi$ from~$G$ to~$H$ such that $\phi(v)=w$ if and only if~$H[N_H(w)]$ has a matching on at least~$p$ edges and $d_G(v) \leq d_H(w)$.

Indeed, if such a locally injective homomorphism~$\phi$ exists, then $d_G(v) \leq d_H(w)$ because~$\phi$ is locally injective.
Furthermore, for every edge $xy$ in $G[N_G(v)]$, the homomorphism~$\phi$ maps the vertices~$x$ and~$y$ to adjacent vertices of $H[N_H(w)]$, and since~$\phi$ is locally injective, it cannot map two vertices of~$N_G(v)$ to the same vertex in~$N_H(w)$.
Therefore~$H[N_H(w)]$ must have a matching on at least~$p$ edges.

Now suppose that~$H[N_H(w)]$ has a matching~$M$ on at least~$p$ edges and $d_G(v) \leq d_H(w)$.
For each edge $xy$ in $G[N_G(v)]$, let~$\phi(x)$ and~$\phi(y)$ be the endpoints of an edge in~$M$ (choosing a different edge of~$M$ for each edge~$xy$).
For the remaining vertices $x \in N_G(v)$, assign the remaining vertices of $N_H(w)$ arbitrarily, such that no two vertices of $N_G(v)$ are assigned the same value (this can be done since $d_G(v) \leq d_H(w)$).
Let $\phi(x)=w$ for all remaining vertices of~$G$ (i.e. the vertex~$v$ and all vertices non-adjacent to~$v$ that have a common neighbour with~$v$).
By construction, $\phi$ is a locally injective homomorphism from~$G$ to~$H$.

The size of a maximum matching in a graph can be found in polynomial time~\cite{Ed65}.
Thus, by branching over the possible vertices~$w \in V(H)$, we obtain a polynomial-time algorithm for \LIHOM{}.
\end{proof}

To prove \NP{}-hardness for results in Lemmas~\ref{lem:lihom-2-2-np} and~\ref{lem:lihom-3-1-np} below, we use a reduction from the \textsc{$H'$-Partition} problem when $H'=P_3$
(the $3$-vertex path) or~$K_3$ (the $3$-vertex complete graph), respectively.
Let~$H'$ be a fixed graph on~$h$ vertices.
The \textsc{$H'$-Partition} problem takes as input a graph~$G'$ on $hn$ vertices and the task is to decide whether the vertex set of~$G'$ can be partitioned into sets $V_1,\ldots,V_n$, each of size~$h$, such that~$G'[V_i]$ contains~$H'$ as a subgraph for all $i \in \{1,\ldots,n\}$.
This problem is known to be \NP{}-complete if $H' \in \{K_3,P_3\}$~\cite{GJ79,KP78}.

\begin{lemma}\label{lem:lihom-2-2-np}
  For $c \geq 2$ and $k\geq 2$, \LIHOM{} is \NP{}-hard on graphs $G$ with $c$-deletion set number $k$.
\end{lemma}

\begin{proof}
We first consider the case when $k=2$.
Consider an instance~$G'$ of the {\sc $P_3$-Partition} problem on~$3n$ vertices, where $n\geq c$.
We construct a graph~$G$ as follows.
For $i \in \{1,\ldots,n\}$, add vertices $a_i,b_i,c_i$ and~$d_i$ and edges $a_ib_i$ and $c_id_i$.
Then add vertices~$u$ and~$v$ and make~$u$ adjacent to $a_i$, $b_i$ and~$d_i$ and~$v$ adjacent to $a_i$, $c_i$ and~$d_i$ for all $i \in \{1,\ldots,n\}$.
Finally, add the edge~$uv$.
Note that $\{u,v\}$ is a minimum-size $c$-deletion set for~$G$ since $\deg_G(u)=\deg_G(v)>c$.
Now let~$H$ be the graph obtained from~$G'$ by adding two vertices~$u'$ and~$v'$ that are adjacent to all the vertices in~$V(G')$ and to each other.
We claim that there is a locally injective homomorphism~$\phi$ from~$G$ to~$H$ if and only if~$G'$ is a yes-instance of the {\sc $P_3$-Partition} problem.

Suppose that~$G'$ is a yes-instance of the {\sc $P_3$-Partition} problem and, for $i \in \{1,\ldots,n\}$, let $v_i^1,v_i^2,v_i^3$ be the three vertices in~$V_i$, such that~$v_i^2$ is adjacent to~$v_i^1$ and~$v_i^3$ ($v_i^1$ may or may not be adjacent to~$v_i^3$).
Let $\phi:V(G)\to V(H)$ be the function such that $\phi(u)=u'$, $\phi(v)=v'$, and for $i \in \{1,\ldots,n\}$, $\phi(a_i)=v_i^1$, $\phi(b_i)=\phi(c_i)=v_i^2$ and $\phi(d_i)=v_i^3$.
Then~$\phi$ is a locally injective homomorphism from~$G$ to~$H$.

Now suppose that~$\phi$ is a locally injective homomorphism from~$G$ to~$H$.
Now $deg_G(u)=deg_G(v)=3n+1$.
Since~$H$ has $3n+2$ vertices and~$\phi$ is a locally injective homomorphism, it follows that~$\phi(u)$ and~$\phi(v)$ must be universal vertices in~$H$.
By symmetry, we may therefore assume that $\phi(u)=u'$ and $\phi(v)=v'$.
Now~$u$ is adjacent to $v$ and the vertices $a_i$, $b_i$ and $d_i$ for all $i \in \{1,\ldots,n\}$.
Similarly, $v$ is adjacent to $u$ and the vertices $a_i$, $c_i$ and $d_i$ for all $i \in \{1,\ldots,n\}$.
Since $deg_G(u)=3n+1$, and~$\phi$ is locally injective, it follows that $\phi(\{a_i,b_i,d_i \;|\; i \in \{1,\ldots,n\}\})=V(G')$.
Similarly, since $deg_G(v)=3n+1$, it follows that $\phi(\{a_i,c_i,d_i \;|\; i \in \{1,\ldots,n\}\})=V(G')$.
Therefore $\phi(\{b_i \;|\; i \in \{1,\ldots,n\}\})=\phi(\{c_i \;|\; i \in \{1,\ldots,n\}\})$.
Renumbering the indices of the~$c_i$ and~$d_i$ vertices if necessary, we may therefore assume by symmetry that $\phi(b_i)=\phi(c_i)$ for all $i \in \{1,\ldots,n\}$.
Now, for all $i \in \{1,\ldots,n\}$, the vertices $a_i$ and~$b_i$ are adjacent in~$G$, so $\phi(a_i)$ and~$\phi(b_i)$ are adjacent in~$H$.
Furthermore the vertices $c_i$ and~$d_i$ are adjacent in~$G$, so $\phi(c_i)=\phi(b_i)$ and~$\phi(d_i)$ are adjacent in~$H$.
We now set $V_i=\{\phi(a_i),\phi(b_i),\phi(d_i)\}$ and note that the~$V_i$ sets partition~$V(G')$, and that~$G'[V_i]$ contains a~$P_3$ subgraph for all $i \in \{1,\ldots,n\}$.
This completes the proof of the case when $k=2$.

To extend the proof to graphs with $c$-deletion number $k>2$, we add $(k-1)$ universal vertices to~$H$ and replace~$u$ with a $k$-clique~$K$ each of whose vertices is adjacent to~$v$ and $a_i$, $b_i$ and~$d_i$ for all $i \in \{1,\ldots,n\}$.
\end{proof}

\begin{lemma}\label{lem:lihom-3-1-np}
  For $c \geq 3$ and $k\geq 1$, \LIHOM{} is \NP{}-hard on graphs $G$ with $c$-deletion set number~$k$.
\end{lemma}

\begin{proof}
We first consider the case when $k=1$.
Consider an instance~$G'$ of the {\sc $K_3$-Partition} problem on~$3n$ vertices, where $n\geq c$.
Let~$H$ be the graph obtained from~$G'$ by adding a universal vertex~$w$.
Let~$G$ be the graph obtained by taking the disjoint union of~$n$ copies of~$K_3$ and adding a universal vertex~$v$.
Note that $\{v\}$ forms a minimum-size $c$-deletion set for~$G$ since $\deg_G(v)>c$.
We claim that there is a locally injective homomorphism~$\phi$ from~$G$ to~$H$ if and only if~$G'$ is a yes-instance of the {\sc $K_3$-Partition} problem.

Indeed, suppose there is such a~$\phi$.
Since~$\phi$ is locally injective and the graphs~$G$ and~$H$ each have~$3n$ vertices, the universal vertex~$v$ must be mapped to a universal vertex of~$H$; without loss of generality, we may therefore assume that $\phi(v)=w$.
Since~$v$ and~$w$ are universal vertices of the same degree, it follows that~$\phi$ is a bijection from~$V(G)$ to~$V(H)$.
Every~$K_3$ in the disjoint union part of~$G$ must therefore be mapped to a~$K_3$ in $H \setminus \{w\}=G'$.
Therefore~$G'$ is a yes-instance of the {\sc $K_3$-Partition} problem.

Now suppose that~$G'$ is a yes-instance of the {\sc $K_3$-Partition} problem.
We let $\phi(v)=w$, and map the vertices of each~$K_3$ in the disjoint union part of~$G$ to some~$V_i$ from the $K_3$-partition of~$H$, mapping each $K_3$ to a different set~$V_i$.
Clearly this is a locally injective homomorphism.
This completes the proof of the case when $k=1$.
To extend the proof to graphs with $c$-deletion number $k>1$, we add $(k-1)$ universal vertices to~$G$ and~$H$.
\end{proof}
\fi

\section{Bounded Tree-depth and Feedback Vertex Set Number}\label{s-npcom}

By Theorem~\ref{the:IHdic}, we already obtained \paraNP{}\hy hardness for \LIHOM{} parameterized by tree-depth or feedback vertex set number.
In this section we show that our tractability results for \LSHOM{} and \LBHOM{} cannot
be significantly extended, since both problems become \paraNP{}\hy
hard parameterized by tree-depth\iflong. \fi
\ifshort as well as parameterized by feedback
vertex set number. Our result strengthens the corresponding result
in~\cite{CFHPT15} for pathwidth and uses a simple adaptation of their
reduction.\fi
\iflong Furthermore, the reduction we give here also provides \paraNP{}\hy hardness for both \LSHOM{} and \LBHOM{} parameterized by the feedback vertex set number. We show this by replacing cycles with stars in the reduction provided in~\cite{CFHPT15} for path-width.
This strengthens their result from path-width to tree-depth and feedback vertex set number.\fi
\iflong \begin{theorem} \fi \ifshort \begin{theorem}[$\star$] \fi\label{t-np}
  \LBHOM{},
  or more specifically, $3$-{\sc FoldCover},
 and \LSHOM{}
  are \NP{}\hy complete on input pairs $(G,H)$ where $G$ has tree-depth at most~$6$ and $H$ has tree-depth at most~$4$.
\end{theorem}
\iflong
\begin{proof}
First note that \LBHOM{}, $3$-{\sc FoldCover} and \LSHOM{} are in \NP{}. 
To prove \NP{}-hardness for $3$-{\sc FoldCover} and \LSHOM{}
we use a reduction from the {\sc $3$-Partition} problem. 
This problem takes as input a multiset $A$ of $3m$ integers, denoted in what follows by $\{a_1,a_2,\ldots,a_{3m}\}$,
and a positive integer $b$, such that $\frac{b}4<a_i<\frac{b}2$ for all $i\in \{1,\ldots,3m\}$ and $\sum_{1\leq i\leq 3m} a_i=mb$. 
The task is to determine whether $A$ can be partitioned into $m$ disjoint sets $A_1,\ldots,A_m$ such that $\sum_{a\in A_i} a=b$ 
for all $i\in \{1,\ldots,m\}$. Note that the restrictions on the size of each element in $A$ implies that each set $A_i$ 
in the desired partition must contain exactly three elements, which is why such a partition $A_1,\ldots,A_m$ is 
called a {\em $3$-partition} of $A$. The {\sc $3$-Partition} problem is strongly \NP{}-complete~\cite{GJ79}, 
i.e. it remains \NP{}-complete even if the problem is encoded in unary. 

We first prove \NP{}-hardness for $3$-{\sc FoldCover}, which implies \NP{}-hardness for \LBHOM.
Given an instance $(A,b)$ of {\sc $3$-Partition}, we construct
an instance of $3$-{\sc FoldCover} consisting of connected graphs $G$ and $H$ with $|V(G)|=3|V(H)|$ as follows. 
To  construct $G$ we take $3m$ disjoint copies $S_1,\ldots,S_{3m}$ of $K_{1,b}$ (stars), one for each element of $A$. 
For each $i\in \{1,\ldots,3m\}$, the vertices of $S_i$ are labelled $c^i,u^i_{1},\ldots,u^i_b$, where $c_i$ is the vertex of degree $b$ in $S_i$ (the centre of the star). We  add 
two new vertices $p^i_j$ and $q^i_j$ for each $i\in \{1,\ldots,3m\}$, $j\in \{1,\ldots,b\}$, as well as two new edges $u^i_jp^i_j$ and $u^i_jq^i_j$. 
We then add three new vertices $x$, $y$ and $z$. The vertex $x$ is made adjacent to the vertices $p^i_1,p^i_2\ldots,p^i_{a_i}$ 
and $q^i_1,q^i_2\ldots,q^i_{a_i}$ for every $i\in \{1,\ldots,3m\}$. 
Finally, the vertex $y$ is made adjacent to every vertex $p^i_j$ that is not adjacent to $x$, and the vertex $z$ is made adjacent 
to every vertex $q^i_j$ that is not adjacent to $x$. This completes the construction of $G$. 
Note that $|V(G)|=3|V(H)|$. 
For an example see Figure~\ref{fig:red}.

To construct $H$, we take $m$ disjoint copies  $\tilde{S}_1,\ldots,\tilde{S}_{m}$ of $K_{1,b}$, 
where the vertices of each star $\tilde{S}_i$ are labelled $\tilde{c}^i,\tilde{u}^i_1,\ldots,\tilde{u}^i_b$. 
For each $i\in \{1,\ldots,m\}$ and $j\in \{1,\ldots,b\}$, we add two vertices $\tilde{p}^i_j$ 
and $\tilde{q}^i_j$ and make both of them adjacent to $\tilde{u}^i_j$. 
Finally, we add a vertex $\tilde{x}$ and make it adjacent to each of the vertices $\tilde{p}^i_j$ and $\tilde{q}^i_j$. 
This finishes the construction of $H$. For an illustration see Figure~\ref{fig:red}.

We now show that there exists a locally bijective homomorphism from $G$ to $H$ if and only if $(A,b)$ is a yes-instance of {\sc $3$-Partition}. 
Let us first assume that there exists a locally bijective homomorphism $\phi$ from $G$ to $H$. 
Since $\phi$ is a degree-preserving mapping, we must have $\phi(x)=\tilde{x}$. 
Moreover, since $\phi$ is locally bijective, the restriction of $\phi$ to $N_G(x)$ is a bijection from $N_G(x)$ to $N_H(\tilde{x})$. 
Again using the definition of a locally bijective mapping, this time considering the neighbourhoods of the vertices in $N_H(\tilde{x})$, 
we deduce that there is a bijection from the set 
$N^2_G(x):= \{u^i_j \mid 1\leq i\leq 3m, 1\le j \le a_i \}$, i.e. from the set of vertices in $G$ at distance $2$ from $x$, to the set 
$N^2_H(\tilde{x}):=\{\tilde{u}^k_j \mid 1\leq k\leq m, 1\le j \le b \}$ 
of vertices that are at distance $2$ from $\tilde{x}$ in $H$. 
For every $k\in \{1,\ldots, m\}$, we define a set $A_k\subseteq A$ such that $A_k$ contains element $a_i\in A$ if and only if 
$\phi(u^i_1)\in \{\tilde{u}^k_1,\ldots,\tilde{u}^k_b\}$. 
Since $\phi$ is a bijection from $N^2_G(x)$ to $N^2_H(\tilde{x})$, the sets $A_1,\ldots,A_m$ are disjoint; 
moreover each element $a_i\in A$ is contained in exactly one of them. 
Since $\phi$ is degree preserving, each $c^i$ has to be mapped onto a $\tilde{c}^j$ (in the special case when $b=3$, we can argue this using the distance to $x$ as before). Additionally, since $\phi$ is locally bijective for every $i\in \{1,\dots,3m\}$, there is a bijection from $N_G(c^i)=\{u^i_1,\dots,u^i_b\}$ to $N_H(\tilde{c}^j)=\{\tilde{u}^j_1,\dots,\tilde{u}^j_b\}$ for the $j\in \{1,\dots,m\}$ for which $\phi(c^i)=\tilde{c}^j$. Combining this and the previous argument  implies that $\sum_{a\in A_i} a=b$ for all $i\in \{1,\ldots,m\}$. Hence $A_1,\ldots,A_m$ is a $3$-partition of $A$.

\begin{figure}
    \centering
    \includegraphics{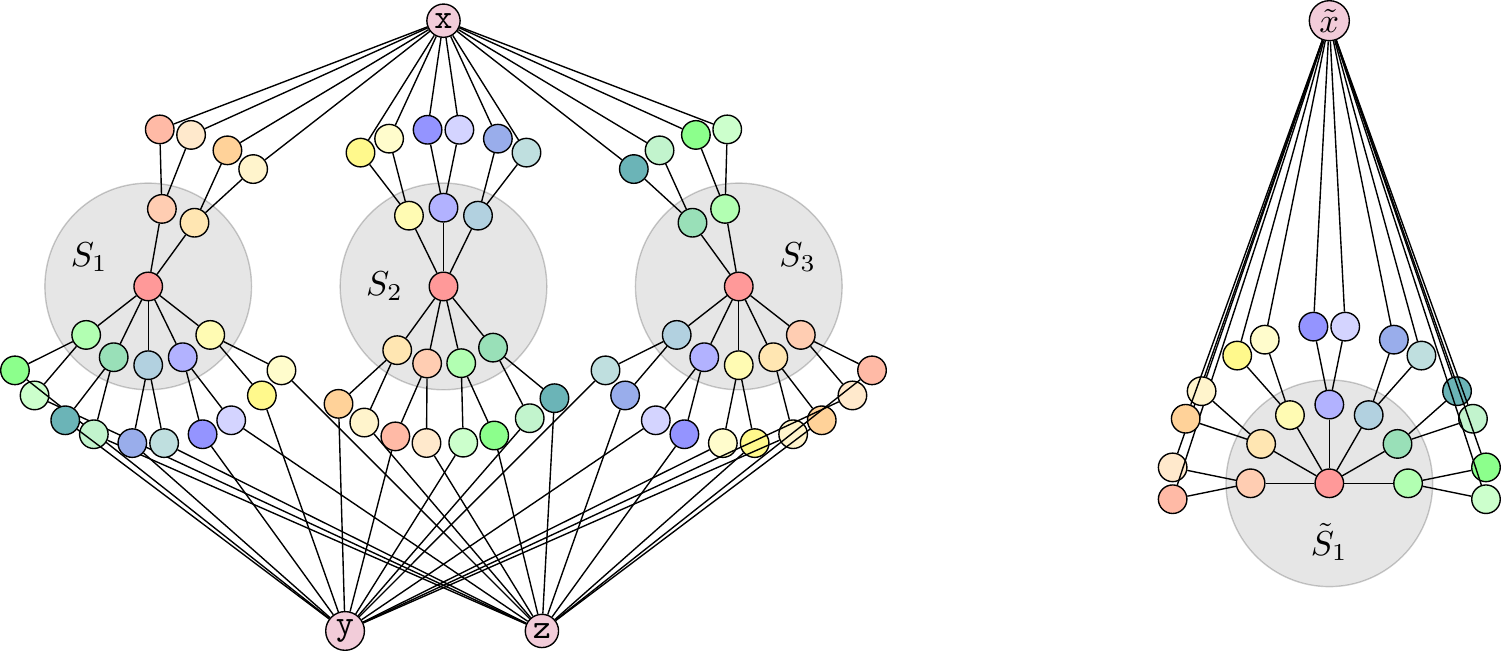}
    
    \caption{An instance of \LBHOM{} consisting of the graph $G$ (left) and the graph $H$ (right) corresponding to the instance $(A,b)$ of {\sc $3$-Partition}, where $A=\{2,3,2\}$ and $b=7$. As $(A,b)$ is a yes-instance of the {\sc $3$-Partition} problem, there is a locally bijective homomorphism from $G$ to $H$ which is indicated by colours.}
    \label{fig:red}
\end{figure}

For the reverse direction, suppose there exists a $3$-partition $A_1,\ldots,A_m$ of $A$. We define a mapping $\phi$ as follows. 
We first set $\phi(x)=\phi(y)=\phi(z)=\tilde{x}$. Let $A_i=\{a_r,a_s,a_t\}$ be any set of the $3$-partition. 
 We  map the vertices of $S_r,S_s,S_t$ to the vertices of $\tilde{S}_i$ in the following way:
 $\phi(c_r)=\phi(c_s)=\phi(c_t)=\tilde{c}_i$;
$\phi(u^r_j)=\tilde{u}^i_j$ for each $j\in \{1,\ldots,b\}$; $\phi(u^s_j)=\tilde{u}^i_{a_r+j}$ for each $j\in \{1,\ldots,a_s+a_t\}$ and $\phi(u^s_j)=\tilde{u}^i_{a_r+j-b}$ for $j\in \{a_s+a_t+1,\dots,b\}$; and $\phi(u^t_j)=\tilde{u}^i_{a_r+a_s+j}$ for each $j\in \{1,\ldots,a_t\}$ and $\phi(u^s_j)=\tilde{u}^i_{a_r+j-b}$ for $j\in \{a_t+1,\dots,b\}$. 
It remains to map the vertices $p^i_j$ and $q^i_j$ for each $i\in \{1,\ldots,3m\}$ and $j\in \{1,\ldots,b\}$. 
Let $p^i_j,q^i_j$ be a pair of vertices in $G$ that are adjacent to $x$, and let $u^i_j$ be the second common neighbour of $p^i_j$ and $q^i_j$. Suppose $\tilde{u}^k_\ell$ is the image of $u^i_j$, i.e. suppose that $\phi(u^i_j)=\tilde{u}^k_\ell$. Then we map $p^i_j$ and $q^i_j$ to $\tilde{p}^k_\ell$ and $\tilde{q}^k_\ell$, respectively. 
We now consider the neighbours of $y$ and $z$ in $G$. By construction, the neighbourhood of $y$ consists of the $2mb$ vertices in the set 
$\{p^i_j \mid a_{i+1}\leq j \leq b\}$, while $N_G(z)=\{q^i_j \mid a_{i+1}\leq j \leq b\}$. 

Observe that $\tilde{x}$, the image of $y$ and $z$, is adjacent to two sets of $mb$ vertices: one of the form $\tilde{p}^k_\ell$, the other of the form $\tilde{q}^k_\ell$. 
Hence, we need to map half the neighbours of $y$ to vertices of the form $\tilde{p}^k_\ell$ and half the neighbours of $y$ to vertices of the form $\tilde{q}^k_\ell$ in order 
to make $\phi$ a locally bijective homomorphism. The same should be done with the neighbours of $z$. For every vertex $\tilde{u}^k_\ell$ in $H$, we do as follows. 
By construction, exactly three vertices of $G$ are mapped to $\tilde{u}^k_\ell$, and exactly two of these vertices, say $u^i_j$ and $u^{g}_{h}$, are at distance $2$ from $y$ in $G$. 
We set $\phi(p^i_j)=\tilde{p}^k_\ell$ and $\phi(p^{g}_{h})=\tilde{q}^k_\ell$. 
We also set $\phi(q^i_j)=\tilde{q}^k_\ell$ and $\phi(q^{g}_{h})=\tilde{p}^k_\ell$.
This completes the definition of the mapping $\phi$.
For an illustration of the map $\phi$, see Figure~\ref{fig:red}.

Since the mapping $\phi$ preserves adjacencies, it clearly is a homomorphism. In order to show that $\phi$ is locally bijective, we first observe that the degree of every vertex in $G$ is equal to the degree of its image in $H$, in particular, 
$d_G(x)=d_G(y)=d_G(z)=d_H(\tilde{x})=2mb$. From the above description of $\phi$ we get a bijection between the vertices of $N_H(\tilde{x})$ and the vertices of $N_G(v)$ for each $v\in \{x,y,z\}$. For every vertex $p^i_j$ that is adjacent to $x$ and $u^i_j$ in $G$, its image $\tilde{p}^k_\ell$ is adjacent to the images $\tilde{x}$ of $x$ and $\tilde{u}^k_\ell$ of $u^i_j$. For every vertex $p^i_j$ that is adjacent to $y$ (respectively $z$) and $u^i_j$ in $G$, its image $\tilde{p}^k_\ell$ or $\tilde{q}^k_\ell$ is adjacent to $\tilde{x}$ of $y$ (respectively $z$) and $\tilde{u}^k_\ell$ of $u^i_j$. Hence the restriction of $\phi$ to $N_G(p^i_j)$ is bijective for every $i\in \{1,\ldots,3m\}$ and $j\in \{1,\ldots,b\}$, and the same clearly holds for the restriction of $\phi$ to $N_G(q^i_j)$. The vertices of each star $S_i$ are mapped to the vertices of some star $\tilde{S}_k$ in such a way that the centres are mapped to centres. This, together with the fact that the image $\tilde{u}^k_\ell$ of every vertex $u^i_j$ is adjacent to the images $\tilde{p}^k_\ell$ and $\tilde{q}^k_\ell$ of the neighbours $p^i_j$ and $q^i_j$ of $u^i_j$, shows that the restriction of $\phi$ to $N_G(u^i_j)$ is bijective for every $i\in \{1,\ldots,3m\}$ and $j\in \{1,\ldots,b\}$. Finally, the neighbourhood of $c^i$ is clearly mapped to the neighbourhood of $\phi(c^i)$ for every $i\in \{1,\ldots,3m\}$.
We conclude that $\phi$ is a locally bijective homomorphism from $G$ to $H$.

In order to show that the tree-depth of $G$ is at most $6$, we construct a rooted tree $T$ as follows. We let $z$ be the root of $T$ and add one child $y$. For $y$ we add one child $x$. We construct $3m$ children $c^1,\ldots,c^{3m}$ of $x$. Furthermore, to each $c^i$ we add $b$ children $u^i_1,\ldots,u^i_b$. Finally, each $u^i_j$ gets two children $p^i_j$ and $q^i_j$. It is easy to observe that $G$ is a subgraph of $C(T)$ and since $T$ has depth $5$, this implies $\operatorname{td}(G)\leq 6$. Furthermore, for $H$ we can use a very similar approach. We let $\tilde{T}$ be the tree obtained from $T$ by removing $z,y$ from $T$ and letting $x$ be the root. After renaming the vertices appropriately, $H$ is a subgraph of $C(\tilde{T})$ and hence $\operatorname{td}(H)\leq 4$. This completes the proof for $3$-{\sc FoldCover} and therefore \LBHOM{}.

In order to prove \NP{}-hardness for \LSHOM{} we can use the same reduction as for \LBHOM. For this we can argue that there is a locally bijective homomorphism from $G$ to $H$, for the graphs $G$ and $H$ constructed above, if and only if there is a locally surjective homomorphism from $G$ to $H$. While the one direction is clear, if $G\xrightarrow{B} H$ then $G\xrightarrow{S} H$, for the converse direction we can make use of the following statement due to Kristiansen and Telle~\cite{KT00}:
\begin{description}
    \item[(*)]If $G\xrightarrow{S} H$ and $\drm(G)=\drm(H)$, then $G\xrightarrow{B} H$.
\end{description}
Here $\drm(G)$, $\drm(H)$ refers to the degree refinement matrix of $G$ or $H$ respectively, which is defined as follows.
An {\it equitable partition} of a connected graph $G$ is a partition of its vertex set into blocks  $B_1,\ldots, B_k$ such that every vertex in $B_i$
has the same number $m_{i,j}$ of neighbours in $B_j$. Then $\drm(G)=(m_{i,j})$ for $m_{i,j}$ corresponding to the coarsest equitable partition of $G$.
We can easily observe that 
$$
\drm(G)=\drm(H)=
\begin{pmatrix}
0\; & 0\; & 2mb\; & 0 \\
0 & 0 & 2 & 1 \\
1 & 1 & 0 & 0\\
0 & b & 0 & 0
\end{pmatrix},
$$
corresponding to the equitable partitions $B_1=\{x,y,z\}$, $B_2=\{u^i_j\mid i\in \{1,\ldots,3m\}, j\in \{1,\ldots, b\}\}$, $B_3=\{p^i_j,q^i_j\mid i\in \{1,\ldots,3m\}, j\in \{1,\ldots, b\}\}$ and $B_4=\{c^i\mid i\in \{1,\ldots,3m\}\}$ in $G$ and a similar equitable partition  in $H$. Hence by (*) we find that $G\xrightarrow{B}H$ if and only if $G\xrightarrow{S} H$, completing the proof for \LSHOM{}.
\end{proof}
\fi
\iflong \begin{theorem} \fi \ifshort \begin{theorem}[$\star$] \fi\label{t-np2}
 \LBHOM{}, or more specifically, {\sc $3$-FoldCover}, and \LSHOM{} are \NP{}\hy complete on input pairs $(G,H)$ where $G$ and $H$ have feedback vertex set number at most~$3$ and~$1$, respectively.
\end{theorem}
\iflong
\begin{proof}
    To prove the statement we use the same reductions as in the proof of Theorem~\ref{t-np}. This is sufficient, as the set $\{x,y,z\}$ is a feedback vertex set of $G$ and the set $\{\tilde{x}\}$ is a feedback vertex set of $H$ for graphs $G$ and $H$ defined in the proof of Theorem~\ref{t-np}.
\end{proof}
\fi
\section{Conclusions}\label{s-con}

\iflong We introduced a general algorithmic framework that can be employed for a wide variety of problems related to
homomorphisms on graphs that have a small fracture number. We already illustrated the applicability
of the framework for three well-known variants of the locally constrained homomorphism problem, i.e. \LSHOM{}, \LBHOM{}, and \LIHOM{}, as well as the
{\sc Role Assignment} problem, by giving three \FPT{} results and one \XP{} result. Our complementary hardness results provide
a fairly comprehensive picture concerning the parameterized complexity of the three
locally constrained homomorphism problems\iflong\ (see also Table~\ref{t-thetable})\fi.
\fi

For future work we aim to extend our ILP-based framework. If
successful, this will then also enable us to address the parameterized complexity of other graph homomorphism variants such as quasi-covers~\cite{FT13} and pseudo-covers~\cite{Ch05,CP11,CP14}.
We also recall an interesting open problem from~\cite{CFHPT15}.
Namely, are \LBHOM{} and \LSHOM{} in \FPT\ when parameterized by the treewidth of the guest graph plus the maximum degree of the guest graph?


\begin{thebibliography}{10}

\bibitem{AFS91}
James Abello, Michael~R. Fellows, and John Stillwell.
\newblock On the complexity and combinatorics of covering finite complexes.
\newblock {\em Australasian Journal of Combinatorics}, 4:103--112, 1991.

\bibitem{An80}
Dana Angluin.
\newblock Local and global properties in networks of processors (extended
  abstract).
\newblock {\em Proc. STOC 1980}, pages 82--93, 1980.

\bibitem{AG81}
Dana Angluin and A.~Gardiner.
\newblock Finite common coverings of pairs of regular graphs.
\newblock {\em Journal of Combinatorial Theory, Series {B}}, 30:184--187, 1981.

\bibitem{Bi74}
Norman~J. Biggs.
\newblock {\em Algebraic Graph Theory}.
\newblock Cambridge University Press, 1974.

\bibitem{Bi82}
Norman~J. Biggs.
\newblock Constructing $5$-arc transitive cubic graphs.
\newblock {\em Journal of the London Mathematical Society II}, 26:193--200,
  1982.

\bibitem{BLT11}
Ondrej B{\'{\i}}lka, Bernard Lidick{\'{y}}, and Marek Tesar.
\newblock Locally injective homomorphism to the simple weight graphs.
\newblock {\em Proc. {TAMC} 2011, LNCS}, 6648:471--482, 2011.

\bibitem{Bo89}
Hans~L. Bodlaender.
\newblock The classification of coverings of processor networks.
\newblock {\em Journal of Parallel and Distributed Computing}, 6:166--182,
  1989.

\bibitem{BGHK95}
Hans~L. Bodlaender, John~R. Gilbert, Hj{\'{a}}lmtyr Hafsteinsson, and Ton
  Kloks.
\newblock Approximating treewidth, pathwidth, frontsize, and shortest
  elimination tree.
\newblock {\em Journal of Algorithms}, 18:238--255, 1995.

\bibitem{BFHJK}
Jan Bok, Jir{\'{\i}} Fiala, Petr Hlinen{\'{y}}, Nikola Jedlickov{\'{a}}, and
  Jan Kratochv{\'{\i}}l.
\newblock Computational complexity of covering two-vertex multigraphs with
  semi-edges.
\newblock {\em CoRR}, abs/2103.15214, 2021.

\bibitem{BTV13}
Binh{-}Minh Bui{-}Xuan, Jan~Arne Telle, and Martin Vatshelle.
\newblock Fast dynamic programming for locally checkable vertex subset and
  vertex partitioning problems.
\newblock {\em Theoretical Computer Science}, 511:66--76, 2013.

\bibitem{Ch05}
J{\'{e}}r{\'{e}}mie Chalopin.
\newblock Local computations on closed unlabelled edges: The election problem
  and the naming problem.
\newblock {\em Proc. {SOFSEM} 2005, LNCS}, 3381:82--91, 2005.

\bibitem{CMZ06}
J{\'{e}}r{\'{e}}mie Chalopin, Yves M{\'{e}}tivier, and Wieslaw Zielonka.
\newblock Local computations in graphs: The case of cellular edge local
  computations.
\newblock {\em Fundamenta Informaticae}, 74:85--114, 2006.

\bibitem{CP11}
J{\'{e}}r{\'{e}}mie Chalopin and Dani{\"{e}}l Paulusma.
\newblock Graph labelings derived from models in distributed computing: {A}
  complete complexity classification.
\newblock {\em Networks}, 58:207--231, 2011.

\bibitem{CP14}
J{\'{e}}r{\'{e}}mie Chalopin and Dani{\"{e}}l Paulusma.
\newblock Packing bipartite graphs with covers of complete bipartite graphs.
\newblock {\em Discrete Applied Mathematics}, 168:40--50, 2014.

\bibitem{CFHPT15}
Steven Chaplick, Jir{\'{\i}} Fiala, Pim van~'t Hof, Dani{\"{e}}l Paulusma, and
  Marek Tesar.
\newblock Locally constrained homomorphisms on graphs of bounded treewidth and
  bounded degree.
\newblock {\em Theoretical Computer Science}, 590:86--95, 2015.

\bibitem{CR00}
Chandra Chekuri and Anand Rajaraman.
\newblock Conjunctive query containment revisited.
\newblock {\em Theoretical Computer Science}, 239:211--229, 2000.

\bibitem{CyganFKLMPPS15}
Marek Cygan, Fedor~V. Fomin, Lukasz Kowalik, Daniel Lokshtanov, D{\'{a}}niel
  Marx, Marcin Pilipczuk, Micha{\l} Pilipczuk, and Saket Saurabh.
\newblock {\em Parameterized Algorithms}.
\newblock Springer, 2015.

\bibitem{DKV02}
V{\'{\i}}ctor Dalmau, Phokion~G. Kolaitis, and Moshe~Y. Vardi.
\newblock Constraint satisfaction, bounded treewidth, and finite-variable
  logics.
\newblock {\em Proc. {CP} 2002, LNCS}, 2470:310--326, 2002.

\bibitem{diestel00}
Reinhard Diestel.
\newblock {\em Graph Theory}, volume 173 of {\em Graduate Texts in
  Mathematics}.
\newblock Springer Verlag, New York, 2nd edition, 2000.

\bibitem{Do16}
Mitre~Costa Dourado.
\newblock Computing role assignments of split graphs.
\newblock {\em Theoretical Computer Science}, 635:74--84, 2016.

\bibitem{DF95}
Rodney~G. Downey and Michael~R. Fellows.
\newblock Fixed-parameter tractability and completeness {II:} on completeness
  for {W[1]}.
\newblock {\em Theoretical Computer Science}, 141:109--131, 1995.

\bibitem{DowneyF13}
Rodney~G. Downey and Michael~R. Fellows.
\newblock {\em Fundamentals of Parameterized Complexity}.
\newblock Texts in Computer Science. Springer, 2013.

\bibitem{DDH16}
P{\aa}l~Gr{\o}n{\aa}s Drange, Markus~S. Dregi, and Pim van~'t Hof.
\newblock On the computational complexity of vertex integrity and component
  order connectivity.
\newblock {\em Algorithmica}, 76:1181--1202, 2016.

\bibitem{DEGKO17}
Pavel Dvor{\'{a}}k, Eduard Eiben, Robert Ganian, Dusan Knop, and Sebastian
  Ordyniak.
\newblock Solving integer linear programs with a small number of global
  variables and constraints.
\newblock {\em Proc. {IJCAI} 2017}, pages 607--613, 2017.

\bibitem{Ed65}
Jack Edmonds.
\newblock Paths, trees, and flowers.
\newblock {\em Canadian Journal of Mathematics}, 17:449--467, 1965.

\bibitem{EB91}
Martin~G. Everett and Stephen~P. Borgatti.
\newblock Role colouring a graph.
\newblock {\em Mathematical Social Sciences}, 21:183--188, 1991.

\bibitem{FellowsLokshtanovMisraRS08}
Michael~R. Fellows, Daniel Lokshtanov, Neeldhara Misra, Frances~A. Rosamond,
  and Saket Saurabh.
\newblock Graph layout problems parameterized by vertex cover.
\newblock In {\em ISAAC}, Lecture Notes in Computer Science, pages 294--305.
  Springer, 2008.

\bibitem{FKK01}
Jir{\'{\i}} Fiala, Ton Kloks, and Jan Kratochv{\'{\i}}l.
\newblock Fixed-parameter complexity of lambda-labelings.
\newblock {\em Discrete Applied Mathematics}, 113:59--72, 2001.

\bibitem{FK02}
Jir{\'{\i}} Fiala and Jan Kratochv{\'{\i}}l.
\newblock Partial covers of graphs.
\newblock {\em Discussiones Mathematicae Graph Theory}, 22:89--99, 2002.

\bibitem{FK08}
Jir{\'{\i}} Fiala and Jan Kratochv{\'{\i}}l.
\newblock Locally constrained graph homomorphisms - structure, complexity, and
  applications.
\newblock {\em Computer Science Review}, 2:97--111, 2008.

\bibitem{FKP08}
Jir{\'{\i}} Fiala, Jan Kratochv{\'{\i}}l, and Attila P{\'{o}}r.
\newblock On the computational complexity of partial covers of theta graphs.
\newblock {\em Discrete Applied Mathematics}, 156:1143--1149, 2008.

\bibitem{FP05}
Jir{\'{\i}} Fiala and Dani{\"{e}}l Paulusma.
\newblock A complete complexity classification of the role assignment problem.
\newblock {\em Theoretical Computer Science}, 349:67--81, 2005.

\bibitem{FP10}
Jir{\'{\i}} Fiala and Dani{\"{e}}l Paulusma.
\newblock Comparing universal covers in polynomial time.
\newblock {\em Theory of Computing Systems}, 46:620--635, 2010.

\bibitem{FT13}
Jir{\'{\i}} Fiala and Marek Tesar.
\newblock Dichotomy of the ${H}$-{Quasi}-{Cover} problem.
\newblock {\em Proc. {CSR} 2013, LNCS}, 7913:310--321, 2013.

\bibitem{FlumGrohe06}
J\"{o}rg Flum and Martin Grohe.
\newblock {\em Parameterized Complexity Theory}, volume XIV of {\em Texts in
  Theoretical Computer Science. An EATCS Series}.
\newblock Springer Verlag, Berlin, 2006.

\bibitem{FrankTardos87}
Andr{\'a}s Frank and {\'E}va Tardos.
\newblock An application of simultaneous diophantine approximation in
  combinatorial optimization.
\newblock {\em Combinatorica}, 7(1):49--65, 1987.

\bibitem{Fr90}
Eugene~C. Freuder.
\newblock Complexity of $k$-tree structured constraint satisfaction problems.
\newblock {\em Proc. AAAI 1990}, pages 4--9, 1990.

\bibitem{GJ79}
Michael~Randolph Garey and David~S. Johnson.
\newblock {\em Computers and Intractability: A Guide to the Theory of
  {NP-Completeness}}.
\newblock W. H. Freeman \& Co., New York, NY, USA, 1979.

\bibitem{GK03}
Michael~U. Gerber and Daniel Kobler.
\newblock Algorithms for vertex-partitioning problems on graphs with fixed
  clique-width.
\newblock {\em Theoretical Computer Science}, 299(1-3):719--734, 2003.

\bibitem{Gr07}
Martin Grohe.
\newblock The complexity of homomorphism and constraint satisfaction problems
  seen from the other side.
\newblock {\em Journal of the {ACM}}, 54:1:1--1:24, 2007.

\bibitem{HHP12}
Pinar Heggernes, Pim van~'t Hof, and Dani{\"{e}}l Paulusma.
\newblock Computing role assignments of proper interval graphs in polynomial
  time.
\newblock {\em Journal of Discrete Algorithms}, 14:173--188, 2012.

\bibitem{HN90}
Pavol Hell and Jaroslav Ne\v{s}et\v{r}il.
\newblock On the complexity of $h$-coloring.
\newblock {\em Journal of Combinatorial Theory, Series {B}}, 48:92--110, 1990.

\bibitem{HN92}
Pavol Hell and Jaroslav Ne\v{s}et\v{r}il.
\newblock The core of a graph.
\newblock {\em Discrete Mathematics}, 109:117--126, 1992.

\bibitem{HN04}
Pavol Hell and Jaroslav Ne\v{s}et\v{r}il.
\newblock {\em Graphs and Homomorphisms}.
\newblock Oxford University Press, 2004.

\bibitem{Kannan87}
Ravi Kannan.
\newblock Minkowski's convex body theorem and integer programming.
\newblock {\em Math. Oper. Res.}, 12(3):415--440, 1987.

\bibitem{KP78}
David~G. Kirkpatrick and Pavol Hell.
\newblock On the completeness of a generalized matching problem.
\newblock In {\em Proceedings of the Tenth Annual ACM Symposium on Theory of
  Computing}, STOC '78, page 240–245, New York, NY, USA, 1978. Association
  for Computing Machinery.

\bibitem{Kl17}
Pavel Klav\'ik.
\newblock {\em Extension Properties of Graphs and Structures}.
\newblock PhD thesis, Charles University, Prague, 2017.

\bibitem{Kr94}
Jan Kratochv{\'{\i}}l.
\newblock Regular codes in regular graphs are difficult.
\newblock {\em Discrete Mathematics}, 133:191--205, 1994.

\bibitem{KPT97}
Jan Kratochv{\'{\i}}l, Andrzej Proskurowski, and Jan~Arne Telle.
\newblock Covering regular graphs.
\newblock {\em Journal of Combinatorial Theory, Series {B}}, 71:1--16, 1997.

\bibitem{KPT98}
Jan Kratochv{\'{\i}}l, Andrzej Proskurowski, and Jan~Arne Telle.
\newblock On the complexity of graph covering problems.
\newblock {\em Nordic Journal of Computing}, 5:173--195, 1998.

\bibitem{KTT16}
Jan Kratochv{\'{\i}}l, Jan~Arne Telle, and Marek Tesar.
\newblock Computational complexity of covering three-vertex multigraphs.
\newblock {\em Theoretical Computer Science}, 609:104--117, 2016.

\bibitem{KT00}
Petter Kristiansen and Jan~Arne Telle.
\newblock Generalized ${H}$-coloring of graphs.
\newblock {\em Proc. ISAAC 2000, LNCS}, 1969:456--466, 2000.

\bibitem{DBLP:journals/ai/KroneggerOP19}
Martin Kronegger, Sebastian Ordyniak, and Andreas Pfandler.
\newblock Backdoors to planning.
\newblock {\em Artif. Intell.}, 269:49--75, 2019.

\bibitem{Lenstra83}
Hendrik~W. {Lenstra Jr.}
\newblock Integer programming with a fixed number of variables.
\newblock {\em Math. Oper. Res.}, 8(4):538--548, 1983.

\bibitem{LT10}
Bernard Lidick{\'{y}} and Marek Tesar.
\newblock Complexity of locally injective homomorphism to the theta graphs.
\newblock {\em Proc. {IWOCA} 2010, LNCS}, 6460:326--336, 2010.

\bibitem{LokshtanovPPS17}
Daniel Lokshtanov, Marcin Pilipczuk, Micha{\l} Pilipczuk, and Saket Saurabh.
\newblock Fixed-parameter tractable canonization and isomorphism test for
  graphs of bounded treewidth.
\newblock {\em {SIAM} J. Comput.}, 46(1):161--189, 2017.

\bibitem{Ma67}
William~S. Massey.
\newblock {\em Algebraic Topology: An Introduction}.
\newblock Harcourt, Brace and World, 1967.

\bibitem{Ne71}
Jaroslav Ne\v{s}et\v{r}il.
\newblock Homomorphisms of derivative graphs.
\newblock {\em Discrete Mathematics}, 1:257--268, 1971.

\bibitem{NOdM12}
Jaroslav Ne\v{s}et\v{r}il and Patrice {Ossona de Mendez}.
\newblock {\em {Sparsity: Graphs, Structures, and Algorithms}}, volume~28 of
  {\em Algorithms and Combinatorics}.
\newblock Springer, 2012.

\bibitem{OR20}
Karolina Okrasa and Pawe{\l} Rz{\k a}\.{z}ewski.
\newblock Subexponential algorithms for variants of the homomorphism problem in
  string graphs.
\newblock {\em Journal of Computer and System Sciences}, 109:126--144, 2020.

\bibitem{PRS21}
Sukanya Pandey, Venkatesh Raman, and Vibha Sahlot.
\newblock Parameterizing role coloring on forests.
\newblock {\em Proc. {SOFSEM} 2021, LNCS}, 12607:308--321, 2021.

\bibitem{PS}
Sukanya Pandey and Vibha Sahlot.
\newblock Role coloring bipartite graphs.
\newblock {\em CoRR}, abs/2102.01124, 2021.

\bibitem{PR01}
Aleksandar Peke\v{c} and Fred~S. Roberts.
\newblock The role assignment model nearly fits most social networks.
\newblock {\em Mathematical Social Sciences}, 41:275--293, 2001.

\bibitem{PR15}
Christopher Purcell and M.~Puck Rombach.
\newblock On the complexity of role colouring planar graphs, trees and
  cographs.
\newblock {\em Journal of Discrete Algorithms}, 35:1--8, 2015.

\bibitem{PR}
Christopher Purcell and M.~Puck Rombach.
\newblock Role colouring graphs in hereditary classes.
\newblock {\em CoRR}, abs/1802.10180, 2018.

\bibitem{RS01}
Fred~S. Roberts and Li~Sheng.
\newblock How hard is it to determine if a graph has a $2$-role assignment?
\newblock {\em Networks}, 37:67--73, 2001.

\bibitem{TP97}
Jan~Arne Telle and Andrzej Proskurowski.
\newblock Algorithms for vertex partitioning problems on partial
  \emph{k}-trees.
\newblock {\em {SIAM} Journal on Discrete Mathematics}, 10(4):529--550, 1997.

\bibitem{HPR10}
Pim van~'t Hof, Dani{\"{e}}l Paulusma, and Johan M.~M. van Rooij.
\newblock Computing role assignments of chordal graphs.
\newblock {\em Theoretical Computer Science}, 411:3601--3613, 2010.

\bibitem{WR83}
Douglas~R. White and Karl~P. Reitz.
\newblock Graph and semigroup homomorphisms on networks of relations.
\newblock {\em Social Networks}, 5:193--235, 1983.

\end{thebibliography}
\end{document}